\documentclass[journal,doublecolumn]{IEEEtran}
\usepackage{latexsym,amssymb,amsmath,graphicx,epstopdf,epsf,cite,bbm,float,subfig}
\usepackage{ifpdf}

\AtBeginDocument{
\addtolength{\abovedisplayskip}{-1.1ex}
\addtolength{\abovedisplayshortskip}{-1.1ex}
\addtolength{\belowdisplayskip}{-1.1ex}
\addtolength{\belowdisplayshortskip}{-1.1ex}
\addtolength{\belowcaptionskip}{-4ex}
}
\newcommand{\argmax}{\operatornamewithlimits{arg\ max}}

\def\ninept{\def\baselinestretch{0.95}}
\ninept

\newcommand{\bx}{{\mathbf{x}}}
\newcommand{\bu}{{\mathbf{u}}}
\newcommand{\bv}{{\mathbf{v}}}
\newcommand{\bth}{{\mathbf{\theta}}}

\newcommand{\bh}{{\mathbf{h}}}

\newcommand{\bp}{{\mathbf{p}}}
\newcommand{\mmin}{{\mathrm{min}}}
\newcommand{\mmax}{{\mathrm{max}}}
\newcommand{\bX}{{\mathbf{X}}}

\newcommand{\bB}{{\mathbf{B}}}

\newcommand{\bQ}{{\mathbf{Q}}}
\newcommand{\mP}{{\mathbf{P}}}
\newcommand{\bLambda}{{\mathbf{\Lambda}}}
\newcommand{\bb}{{\mathbf{b}}}

\newcommand{\bw}{{\mathbf{w}}}
\newcommand{\bee}{{\mathbf{e}}}
\newcommand{\cX}{{\mathcal{X}}}

\newcommand{\cS}{{\mathcal{S}}}
\newcommand{\cB}{{\mathcal{B}}}

\newcommand{\cH}{{\mathcal{H}}}
\newcommand{\cZ}{{\mathcal{Z}}}
\newcommand{\cO}{{\mathcal{O}}}
\newcommand{\cU}{{\mathcal{U}}}
\newcommand{\cM}{{\mathcal{M}}}

\newcommand{\cV}{{\mathcal{V}}}
\newcommand{\cN}{{\mathcal{N}}}
\newcommand{\cR}{{\mathcal{R}}}
\newcommand{\cP}{{\mathcal{P}}}
\newcommand{\bP}{{\mathrm{Pr}}}

\renewcommand{\vec}[1]{\mbox{\boldmath${#1}$}}

\newcommand{\limn}{\lim\limits_{n \rightarrow \infty}}

\newcommand{\vpi}{\vec{\pi}}
\newcommand{\ei}{\end{itemize}}
\newcommand{\bi}{\begin{itemize}}

\newcommand{\vb}{\mbox{$\vec{b}$}}

\newtheorem{theorem}{\bf{Theorem}}[section]

\newtheorem{remark}{\bf{Remark}}[section]
\newtheorem{lemma}{\bf{Lemma}}[section]
\newtheorem{corollary}{\bf{Corollary}}[section]
\newcommand{\defi}{\stackrel{\bigtriangleup}{=}}

    \def\squarebox#1{\hbox to #1{\hfill\vbox to #1{\vfill}}}

\newcommand{\eps}{\mbox{$\epsilon$}}
\newcommand{\beps}{\mathbf{\mbox{$\epsilon$}}}

\begin{document}
\title{Optimal Investment Under Transaction Costs} \vspace{-0.2in}
 \author{S.~Tunc, M.~A.~Donmez and
  S.~S.~Kozat, {\em Senior Member}, IEEE\thanks{S.~S.~Kozat, S.~Tunc and
    M.~A.~Donmez (\{skozat,saittunc,medonmez\}@ku.edu.tr) are
    with the EE Department at the Koc University, Sariyer, 34450,
    Istanbul, tel: +902123381864, fax: +902123381548. This work is
    supported in part by the IBM Faculty Award and the Outstanding Young
    Scientist Award Program.}}  \maketitle
\vspace{-0.2in}
\begin{abstract}
 We investigate how and when to diversify capital over assets, i.e.,
 the portfolio selection problem, from a signal processing
 perspective. To this end, we first construct portfolios that achieve
 the optimal expected growth in i.i.d. discrete-time two-asset markets
 under proportional transaction costs. We then extend our analysis to
 cover markets having more than two stocks. The market is modeled by
 a sequence of price relative vectors with arbitrary discrete
 distributions, which can also be used to approximate a wide class of
 continuous distributions. To achieve the optimal growth, we use threshold
 portfolios, where we introduce a recursive update to calculate the
 expected wealth.  We then demonstrate that under the threshold
 rebalancing framework, the achievable set of portfolios elegantly
 form an irreducible Markov chain under mild technical conditions. We
 evaluate the corresponding stationary distribution of this Markov
 chain, which provides a natural and efficient method to calculate the
 cumulative expected wealth. Subsequently, the corresponding
 parameters are optimized yielding the growth optimal portfolio under
 proportional transaction costs in i.i.d. discrete-time two-asset
 markets.  As a widely known financial problem, we next solve optimal
 portfolio selection in discrete-time markets constructed by sampling
 continuous-time Brownian markets. For the case that the underlying
 discrete distributions of the price relative vectors are unknown, we
 provide a maximum likelihood estimator that is also incorporated in
 the optimization framework in our simulations.
\end{abstract}
\vspace{-0.2in}
\section{Introduction \label{sec:introduction}}
The problem of how and when an investor should diversify capital over
various assets, whose future returns are yet to be realized, is
extensively studied in various different fields from financial
engineering \cite{markowitz1952,markowitz1991}, signal processing
\cite{bean1, bean2, akansu,flier1, flier2}, machine learning
\cite{helmb98,vovk98} to information theory
\cite{cover1996}. Naturally this is one of the most important
financial applications due to the amount of money involved. However,
the current financial crisis demonstrated that there is a significant
room for improvement in this field by sound signal processing methods
\cite{flier1,flier2}, which is the main goal of this paper. In this
paper, we investigate how and when to diversify capital over assets,
i.e., the portfolio selection problem, from a signal processing
perspective and provide portfolio selection strategies that maximize
the expected cumulative wealth in discrete-time markets under
proportional transaction costs.

In particular, we study an investment problem in markets that allow
trading at discrete periods, where the discrete period is arbitrary,
e.g., it can be seconds, minutes or days \cite{investment}.
Furthermore the market levies transaction fees for both selling and
buying an asset proportional to the volume of trading at each
transaction, which accurately models a broad range of financial
markets \cite{investment, blum1998}. In our discussions, we first
consider markets with two assets, i.e., two-asset markets. We
emphasize that the two-stock markets are extensively studied in
financial literature and are shown to accurately model a wide range of
financial applications \cite{investment} such as the well-known
``Stock and Bond Market" \cite{investment}. We then extend our
analysis to markets having more than two assets, i.e., $m$-stock
markets, where $m$ is arbitrary.  Following the extensive literature
\cite{cover1996, KoSi11, markowitz1952,markowitz1991, akansu,
  investment}, the market is modeled by a sequence of price relative
vectors, say $\{ \vec{X}(n) \}_{n \geq 1}$, $\vec{X}(n) \in
[0,\infty)^m$, where each entry of $\vec{X}(n)$, i.e., $X_i(n) \in
  [0,\infty)$, is the ratio of the closing price to the opening price
    of the $i$th stock per investment period. Hence, each entry of
    $\vec{X}(n)$ quantifies the gain (or the loss) of that asset at
    each investment period. The sequence of price relative vectors is
    assumed to have an i.i.d. ``discrete'' distribution
    \cite{markowitz1952,markowitz1991, akansu, investment}, however,
    the discrete distributions on the vector of price relatives are
    arbitrary. In this sense, the corresponding discrete distributions
    can approximate a wide class of continuous distributions on the
    price relatives that satisfy certain regularity conditions by
    appropriately increasing the size of the discrete sample space. We
    first assume that we know the discrete distributions on the price
    relative vectors and then extend our analysis to cover when the
    underlying distributions are unknown. We emphasize that the
    i.i.d. assumption on the sequence of price relative vectors is
    shown to hold in most realistic markets \cite{investment,
      iyengargo}. The detailed market model is provided in Section IV.
    At each investment period, the diversification of the capital over
    the assets is represented by a portfolio vector $\vec{b}(n)$,
    where for each entry $1\geq b_i(n) \geq 0$, $\sum_{i=1}^m
    b_i(n)=1$, and $b_i(n)$ is the ratio of the capital invested in
    the $i$th asset at investment period $n$. As an example if we
    invest using $\vec{b}(n)$, we earn (or loose)
    $\vec{b}^T(n)\vec{X}(n)$ at the investment period $n$ after
    $\vec{X}(n)$ is revealed. Given that we start with one dollars,
    after an investment period of $N$ days, we have the wealth growth
    $\prod_{n=1}^N \vec{b}^T(n)\vec{X}(n)$. Under this general market
    model, we provide algorithms that maximize the expected growth
    over any period $N$ by using ``threshold rebalanced portfolios''
    (TRP)s, which are shown to yield optimal growth in general
    i.i.d. discrete-time markets \cite{iyengargo}.

	\vspace{-0.09in}
Under mild assumptions on the sequence of price relatives and without
any transaction costs, Cover et. al \cite{cover1996} showed that the
portfolio that achieves the maximal growth is a constant rebalanced
portfolio (CRP) in i.i.d. discrete-time markets. A CRP is a portfolio
investment strategy where the fraction of wealth invested in each
stock is kept constant at each investment period. A problem
extensively studied in this framework is to find sequential portfolios
that asymptotically achieve the wealth of the best CRP tuned to the
underlying sequence of price relatives.  This amounts to finding a
daily trading strategy that has the ability to perform as well as the
best asset diversified, constantly rebalanced portfolio. Several
sequential algorithms are introduced that achieve the performance of
the best CRP either with different convergence rates or performance on
historical data sets \cite{helmb98,cover1996,vovk98,
  kalai00efficient}.  Even under transaction costs, sequential
algorithms are introduced that achieve the performance of the best CRP
\cite{blum1998}. Nevertheless, we emphasize that keeping a CRP may
require extensive trading due to possible rebalancing at each
investment period deeming CRPs, or even the best CRP, ineffective in
realistic markets even under mild transaction costs \cite{KoSi11}.

In continuous-time markets, however, it has been shown that under
transaction costs, the optimal portfolios that achieve the maximal
wealth are certain class of ``no-trade zone'' portfolios
\cite{davis1990,taksar,BrClVa05}. In simple terms, a no-trade zone
portfolio has a compact closed set such that the rebalancing occurs if
the current portfolio breaches this set, otherwise no rebalancing
occurs. Clearly, such a no-trade zone portfolio may avoid hefty
transaction costs since it can limit excessive rebalancing by defining
appropriate no-trade zones. Analogous to continuous time markets, it
has been shown in \cite{iyengargo} that in two-asset i.i.d. markets
under proportional transaction costs, compact no-trade zone portfolios
are optimal such that they achieve the maximal growth under mild
assumptions on the sequence of price relatives.  In two-asset markets,
the compact no trade zone is represented by thresholds, e.g., if at
investment period $n$, the portfolio is given by $\vb(n) = [b(n)\;\;
  (1-b(n))]^T$, where $1 \geq b(n) \geq 0$, then rebalancing occurs if
$b(n) \notin (\alpha,\beta)$, given the thresholds $\alpha$, $\beta$,
where $1\geq \beta \geq \alpha \geq 0$. Similarly, the interval
$(\alpha,\beta)$ can be represented using a target portfolio $b$ and a
region around it, i.e., $(b-\epsilon,b+\epsilon)$, where $\min
\{b,1-b\} \geq \epsilon \geq 0$ such that $\alpha = b-\epsilon$ and
$\beta = b+\epsilon$. Extension of TRPs to markets having more than
two stocks is straightforward and explained in Section~\ref{subsec:m-asset}.
%

However, how to construct the no-trade zone portfolio, i.e., selecting
the thresholds that achieve the maximal growth, has not yet been
solved except in elementary scenarios \cite{iyengargo}. We emphasize
that a sequential universal algorithm that asymptotically achieves the
performance of the best TRP specifically tuned to the underlying
sequence of price relatives is introduced in \cite{iyengar2005}. This
algorithm leverages Bayesian type weighting from \cite{cover1996}
inspired from universal source coding and requires no statistical
assumptions on the sequence of price relatives. In similar lines,
various different universal sequential algorithms are introduced that
achieve the performance of the best algorithm in different competition
classes in \cite{port, KoSi11, kozat_tree, bean1, bean2, KoSi08}. However, we emphasize that the performance
guarantees in \cite{iyengar2005} (and in \cite{port, KoSi11, bean1,
  bean2}) on the performance, although without
any stochastic assumptions, is given for the worst case sequence and
only optimal in the asymptotic. For any finite investment period, the
corresponding order terms in the upper bounds may not be negligible in
financial markets, although they may be neglected in source coding
applications (where these algorithms are inspired from). We
demonstrate that our algorithm readily outperforms these universal
algorithms over historical data \cite{cover1996}, where similar
observations are reported in \cite{borodin, kozat_tree}.

Our main contributions are as follows.  We first consider two-asset
markets and recursively evaluate the expected achieved wealth of a
threshold portfolio for any $b$ and $\epsilon$ over any investment
period. We then extend this analysis to markets having more than
two-stocks. We next demonstrate that under the threshold rebalancing
framework, the achievable set of portfolios form an irreducible Markov
chain under mild technical conditions. We evaluate the corresponding
stationary distribution of this Markov chain, which provides a natural
and efficient method to calculate the cumulative expected
wealth. Subsequently, the corresponding parameters are optimized using
a brute force approach yielding {\em the growth optimal investment
  portfolio under proportional transaction costs in
  i.i.d. discrete-time} two-asset markets. We note that for the case
with the irreducible Markov chain, which covers practically all
scenarios in the realistic markets, the optimization of the parameters
is offline and carried out only once. However, for the case with
recursive calculations, the algorithm has an exponential computational
complexity in terms of the number of states. However, in our
simulations, we observe that a reduced complexity form of the
recursive algorithm that keeps only a constant number of states by
appropriately pruning certain states provides nearly identical results
with the ``optimal'' algorithm. Furthermore, as a well studied
problem, we also solve optimal portfolio selection in discrete-time
markets constructed by sampling continuous-time Brownian markets
\cite{investment}. When the underlying discrete distributions of the
price relative vectors are unknown, we provide a maximum likelihood
estimator to estimate the corresponding distributions that is
incorporated in the optimization framework in the Simulations
section. For all these approaches, we also provide the corresponding
complexity bounds.

The organization of the paper is as follows. In
Section~\ref{sec:prob}, we briefly describe our discrete-time stock
market model with discrete price relatives and symmetric proportional
transaction costs. In Section~\ref{sec:TRP}, we start to investigate
TRPs, where we first introduce a recursive
update in Section~\ref{subsec:algorithm} for a market having
two-stocks. Generalization of the iterative algorithm to the $m$-asset
market case is provided in Section~\ref{subsec:m-asset}.  We then show
that the TRP framework can be analyzed using finite state Markov
chains in Section~\ref{subsec:finite} and
Section~\ref{subsec:markov}. The special Brownian market is analyzed
in Section~\ref{subsec:brownian}. The maximum likelihood estimator is
derived in Section~\ref{sec:ML}. We simulate the performance of our
algorithms in Section~\ref{sec:sim} and the paper concludes with
certain remarks in Section~\ref{sec:con}.\vspace{-0.1in}

\section{Problem Description}\label{sec:prob}
We consider discrete-time stock markets under transaction costs. We
first consider a market with two stocks and then extend the analysis
to markets having more than two stocks.  We model the market using a
sequence of price relative vectors $\bX(n)$. A vector of price
relatives $\bX(n)=\left[X_1(n),\ldots,X_m(n)\right]^T$ in a market
with $m$ assets represents the change in the prices of the assets over
investment period $n$, i.e., for the $i$th stock $X_i(n)$ is the
ratio of the closing to the opening price of the $i$th stock over
period $n$. For a market having two assets, we have
$\bX(n)=\left[X_1(n)\, X_2(n)\right]^T$. We assume that the price
relative sequences $X_1(n)$ and $X_2(n)$ are independent and
identically distributed (i.i.d.) over with possibly different discrete
sample spaces $\cX_1$ and $\cX_2$, i.e., $X_1(n) \in \cX_1$ and
$X_2(n) \in \cX_2$, respectively \cite{iyengargo}. For technical
reasons, in our derivations, we assume that the sample space is $\cX
\defi \cX_1 \cup \cX_2= \{ x_1, x_2, \ldots, x_K \}$ for both $X_1(n)$
and $X_2(n)$ where $\vert\cX\vert = K$ is the cardinality of the set
$\cX$. The probability mass function (pmf) of $X_1(n)$ is $p_1(x)
\defi \bP(X_1=x)$ and the probability mass function of $X_2(n)$ is
$p_2(x) \defi \bP(X_2=x)$. We define $p_{i,1}=p_1(x_i)$ and
$p_{i,2}=p_2(x_i)$ for $x_i\in\cX$ and the probability mass vectors
$\bp_1 = \left[p_{1,1}\; p_{2,1}\; \ldots\; p_{K,1}\right]^T$ and
$\bp_2 = \left[p_{1,2}\; p_{2,2}\; \ldots\; p_{K,2}\right]^T$,
respectively. Here, we first assume that the corresponding probability
mass vectors $\bp_1$ and $\bp_2$ are known. We then extend our
analysis where $\bp_1$ and $\bp_2$ are unknown and sequentially
estimated using a maximum likelihood estimator in
Section~\ref{sec:ML}.

An allocation of wealth over two stocks is represented by the
portfolio vector $\bb(n) = \left[b(n)\;1-b(n)\right]$, where $b(n)$ and
$1-b(n)$ represents the proportion of wealth invested in the first and
second stocks, respectively, for each investment period $n$. In two
stock markets, the portfolio vector $\bb = \left[b\; 1 - b\right]$ is
completely characterized by the proportion $b$ of the total wealth
invested in the first stock. For notational clarity, we use $b(n)$ to
represent $b_1(n)$ throughout the paper.

We denote a threshold rebalancing portfolio with an initial and target
portfolio $b$ and a threshold $\eps$ by TRP($b$,$\eps$). At each
market period $n$, an investor rebalances the asset allocation only if
the portfolio leaves the interval $(b-\eps,b+\eps)$.  When
$b(n)\not\in(b-\eps,b+\eps)$, the investor buys and sells stocks so
that the asset allocation is rebalanced to the initial allocation,
i.e., $b(n)=b$, and he/she has to pay transaction fees. We emphasize
that the rebalancing can be made directly to the closest boundary
instead of to $b$ as suggested in \cite{iyengargo}, however, we
rebalance to $b$ for notational simplicity and our derivations hold
for that case also. We model transaction cost paid when rebalancing
the asset allocation by a fixed proportional cost $c\in(0,1)$
\cite{blum1998,iyengargo,KoSi11}. For instance, if the investor buys
or sells $S$ dollars of stocks, then he/she pays $cS$ dollars of
transaction fees. Although we assume a symmetric transaction cost
ratio, all the results can be carried over to markets with asymmetric
costs \cite{KoSi11,iyengargo}. Let $S(N)$ denote the achieved wealth
at investment period $N$ and assume, without loss of generality, that
the initial wealth of the investor is 1 dollars. For example, if the
portfolio $b(n)$ does not leave the interval $(b-\eps,b+\eps)$ and the
allocation of wealth is not rebalanced for $N$ investment periods,
then the current proportion of wealth invested in the first stock is
given by $ b(N) = b\prod_{n=1}^N X_1(n)/ \left( b\prod_{n=1}^N X_1(n)+
(1-b)\prod_{n=1}^N X_2(n)\right)$ and achieved wealth is given by $
S(N) = b\prod_{n=1}^N X_1(n) + (1-b)\prod_{n=1}^N X_2(n).$ If the
portfolio leaves the interval $(b-\eps,b+\eps)$ at period $N$, i.e.,
$b(N)\not \in (b-\eps,b+\eps)$, then the investor rebalances the asset
distribution to the initial distribution and pays approximately $S(N)
\vert b(N)-b\vert c$ dollars for transaction costs \cite{blum1998}. In
the next section, we first evaluate the expected achieved wealth
$E[S(N)]$ so that we can optimize $b$ and $\eps$.\vspace{-0.05in}

\section{Threshold Rebalanced Portfolios \label{sec:TRP}}
In this section, we first investigate TRPs
in discrete-time two-asset markets under proportional transaction
costs. We initially calculate the expected achieved wealth at a given
investment period by an iterative algorithm. Then, we present an upper
bound on the complexity of the algorithm. We also calculate the
expected achieved wealth of markets having more than two assets, i.e.,
$m$-asset markets for an arbitrary $m$. We then provide the necessary
and sufficient conditions such that the achievable portfolios are
finite such that the complexity of the algorithm does not grow at any
period. We also show that the portfolio sequence converges to a
stationary distribution and derive the expected achieved wealth. Based
on the calculation of the expected achieved wealth, we optimize $b$
and $\eps$ using a brute-force search. Finally, with these
derivations, we consider the well-known discrete-time two-asset
Brownian market with proportional transaction costs and investigate
the asymptotic expected achieved wealth.\vspace{-0.1in}
\subsection{An Iterative Algorithm \label{subsec:algorithm}}\vspace{-0.05in}
In this section, we calculate the expected wealth growth of a TRP with
an iterative algorithm and find an upper bound on the complexity of
the algorithm. To accomplish this, we first define the set of
achievable portfolios at each investment period since the iterative
calculation of the expected achieved wealth is based on the achievable
portfolio set. We next introduce the portfolio transition sets and
the transition probabilities of achievable portfolios at successive
investment periods in order to find the probability of each portfolio
state iteratively and to calculate $E[S(N)]$. 

\begin{figure}
\centering
\centerline{\epsfxsize=9cm \epsfbox{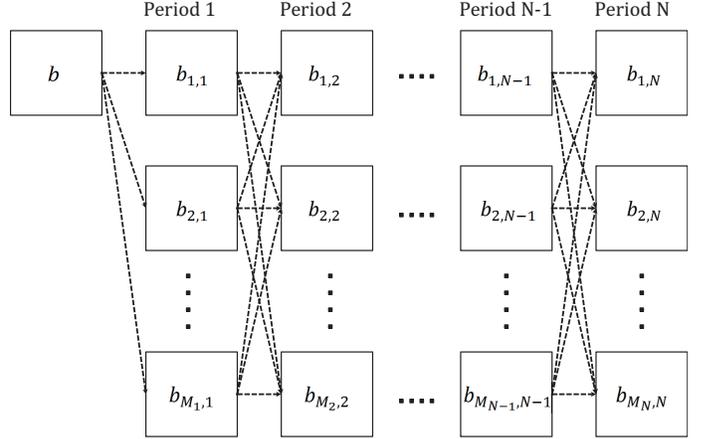}}
\caption{ Block diagram representation of $N$ period investment.}
\label{blockdiagram}
\end{figure}

We define the set of achievable portfolios at each investment period
as follows. Since the sample space of the price relative sequences
$X_1(n)$ and $X_2(n)$ is finite, i.e., $\vert \cX \vert =K$, the set
of achievable portfolios at period $N$ can only have finitely many
elements. We define the set of achievable portfolios at period $N$ as
$\cB_N = \{b_{1,N},\ldots,b_{M_N,N}\}$, where $M_N\defi\vert B_N
\vert$ is the size of the set $\cB_N$ for $N\geq1$. As illustrated in Fig.~\ref{blockdiagram}, for each achievable portfolio $b_{l,N}\in\cB_N$, there is a certain set of portfolios in $\cB_{N-1}$ that are connected to $b_{l,n}$, by definition of $b_{l,n}$. At a given investment period $N$, the set of achievable portfolios $\cB_N$ is given by \vspace{-0.01in}
\begin{align*}
\cB_{N}= & \left\{ b_{1,N},\ldots,b_{M_N,N}\; {\Big \vert}\; b_{l,N}=\frac{b_{k,N-1} u}{b_{k,N-1}u+(1-b_{k,N-1})v} \right. \\ 
& \left. \vphantom{\frac{b_{k,N-1} u}{b_{k,N-1}u+(1-b_{k,N-1})v}} \in(b-\eps,b+\eps)\;\mathrm{or}\;b_{l,N}=b ,\; u,v\in\cX\right\}.
\end{align*}
We let, without loss of generality, $b_{1,N}=b$ for each $N\in\mathbb{N}$. Note that in Fig.~\ref{blockdiagram}, the size of the set of achievable portfolios at each period may grow in the next period depending on the set of price relative vectors. We next define the transition probabilities as $q_{k,l,N}=\bP\left(b(N)=b_{l,N}\vert b(N-1)=b_{k,N-1}\right)$ for $k=1,\ldots,M_{N-1}$ and $l=1,\ldots,M_N$ and the set of achievable portfolios that are connected to $b_{l,N}$, i.e., the portfolio transition set, as $\cN_{l,N}=\{b_{k,N-1} \in \cB_{N-1}\; \vert\; q_{k,l,N}>0,\;k=1,\ldots,M_{N-1}\}$ for $l=1,\ldots,M_N$. Hence, the probability of each portfolio state is given by
\begin{align}
&\bP\left(b(N)=b_{l,N}\right) = \sum_{b_{k,N-1}\in \cB_{N-1}} \bP\left(b(N)=b_{l,N}\vert b(N-1) = \right.\nonumber\\ & \left. b_{k,N-1}\right) \bP\left(b(N-1)=b_{k,N-1}\right) \nonumber \\
& = \sum_{b_{k,N-1}\in \cN_{l,N}} q_{k,l,N} \bP\left(b(N-1)=b_{k,N-1}\right)
\label{eq:eqn1}
\end{align}
for $l=1,\ldots,M_N$. Therefore, we can calculate the probability of achievable portfolios iteratively. Using these iterative equations, we next iteratively calculate the expected achieved wealth $E[S(N)]$ at each period as follows.

By definition of $\cB_{N}$ and using the law of total expectation \cite{woods}, the expected achieved wealth at investment period $N$ can be written as
\begin{equation}
E[S(N)] = \sum_{l=1}^{M_N} \bP\left(b(N)=b_{l,N}\right) E\left[S(N) | b(N)=b_{l,N} \right].
\label{eq:eqn2}
\end{equation} To get $E[S(N)]$ in \eqref{eq:eqn2} iteratively, we evaluate $\bP\left(b(N)=b_{l,N}\right)E\left[S(N) | b(N)=b_{l,N} \right]$ for each $l=1,\ldots,M_N$ from $\bP\left(b(N-1)=b_{k,N-1}\right)E[S(N-1)\vert b(N-1)=b_{k,N-1}]$ for $k=1,\ldots,M_{N-1}$. To achieve this, we first find the transition probabilities (not the state probabilities) between the achievable portfolios.

We define the set of price relative vectors that connect $b_{k,N-1}$ to $b_{l,N}$ as $\cU_{k,l,N}$ where
\begin{align*}
\cU_{k,l,N} &=  \left\{\bw = \left[w_1\;w_2\right]^T \in \cX^2 \; \vphantom{\frac{w_1 b_{k,N-1}}{w_1 b_{k,N-1} + w_2 (1-b_{k,N-1})}}\right.\\ &\left. \;\;\;\;\;{\Big \vert}\; b_{l,N}=\frac{w_1 b_{k,N-1}}{w_1 b_{k,N-1} + w_2 (1-b_{k,N-1})} \right\}\nonumber
\end{align*}for $k=1,\ldots,M_{N-1}$ and $l=2,\ldots,M_N$. We consider the price relative vectors that connect $b_{k,N-1}$ to $b_{1,N}=b$ separately since, in this case, there are two cases depending on whether the portfolio leaves the interval $(b-\eps,b+\eps)$ or not. We define $\cU_{k,1,N}$ as $\cU_{k,1,N} = \cV_{k,1,N} \cup \cR_{k,1,N}$,
where $\cV_{k,1,N}$ is the set of price relative vectors that connect $b_{k,N-1}$ to $b_{1,N}=b$ such that the portfolio does not leave the interval $(b-\eps,b+\eps)$ at period $N$, i.e.,
\begin{align*}
\cV_{k,1,N} &= \left\{\bw = \left[w_1\;w_2\right]^T \in \cX^2 \; \vphantom{\frac{w_1 b_{k,N-1}}{w_1 b_{k,N-1} + w_2 (1-b_{k,N-1})}}\right.\\ &\left. \;\;\;\;\;{\Big \vert}\; \frac{w_1 b_{k,N-1}}{w_1 b_{k,N-1} + w_2 (1-b_{k,N-1})}=b \right\}, \end{align*}
and $\cR_{k,1,N}$ is the set of price relative vectors that connect $b_{k,N-1}$ to $b_{1,N}$ such that the portfolio leaves the interval $(b-\eps,b+\eps)$ at period $N$ and is rebalanced to $b_{1,N}=b$, i.e.,
\begin{align*}
\cR_{k,1,N} = &\bigg\{\bw = \left[w_1\;w_2\right]^T \in \cX^2 \; \vert\; \frac{w_1 b_{k,N-1}}{w_1 b_{k,N-1} + w_2 (1-b_{k,N-1})} \\ &\not\in(b-\eps,b+\eps)\bigg\}.
\end{align*}Then, the transition probabilities are given by
\begin{align}
& q_{k,l,N} = \bP\left(b(N)=b_{l,N}\vert b(N-1)=b_{k,N-1}\right) \nonumber \\
              &= \bP\left(\bX(N) \in \cU_{k,l,N}\right) = \sum_{\bw=\left[w_1\;w_2\right]^T\in \cU_{k,l,N}} p_1(w_1)p_2(w_2) \label{eq:eqn3}
\end{align} for $k=1,\ldots,M_{N-1}$ and $l=1,\ldots,M_N$ so that we can calculate $\bP\left(b(N))=b_{l,N}\right)$ iteratively for each $l=1,\ldots,M_N$ by \eqref{eq:eqn1}. Since we have recursive equations for the state probabilities, we next perform the iterative calculation of the expected achieved wealth based on the achievable portfolio sets and the transition probabilities.

Given the recursive formulation for the state probabilities, we can
evaluate the term \\ $\bP\left(b(N)=b_{l,N}\right)E[S(N)\vert
  b(N)=b_{l,N}]$ for $l=1,\ldots,M_N$ from
$\bP\left(b(N-1)=b_{k,N-1}\right)E[S(N-1)\vert b(N-1)=b_{k,N-1}]$ for
$k=1,\ldots,M_{N-1}$ iteratively to calculate $E[S(N)]$ by
\eqref{eq:eqn2} as follows. To evaluate
$\bP\left(b(N)=b_{l,N}\right)E[S(N)\vert b(N)=b_{l,N}]$, we need to
consider two cases separately based on the value of $b_{l,N}$.

In the first case, we see that if the portfolio $b(N)=b_{l,N}$, where $l=2,\ldots,N$, then the portfolio does not leave the interval $(b-\eps,b+\eps)$ at period $N$. Hence, no transaction cost is paid so that we can express $\bP\left(b(N)=b_{l,N}\right)E[S(N) \vert b(N)=b_{l,N}]$ as a summation of the conditional expectations for all $b_{k,N-1}\in\cN_{l,N}$ by the law of total expectation \cite{woods} as
\begin{align}
&\bP\left(b(N)=b_{l,N}\right)E[S(N) \vert b(N)=b_{l,N}] \nonumber \\
&= \sum_{b_{k,N-1}\in \cN_{l,N}} E\left[S(N) \vert b(N)=b_{l,N},b(N-1)=b_{k,N-1}\right] \nonumber \\ 
& \times\bP\left(b(N-1)=b_{k,N-1}\vert b(N)=b_{l,N}\right)\bP\left(b(N)=b_{l,N}\right) \nonumber \\
&=\sum_{b_{k,N-1}\in \cN_{l,N}} E\left[S(N) \vert b(N)=b_{l,N},b(N-1)=b_{k,N-1}\right] \nonumber \\
& \times\bP\left(b(N-1)=b_{k,N-1}\right)q_{k,l,N} \label{eq:eqn42},
\end{align} where \eqref{eq:eqn42} follows from Bayes' theorem \cite{woods}. We note that given $b(N-1)=b_{k,N-1}$ and $b(N)=b_{l,N}$, the price relative vector $\bX(N)$ can take values from $\cU_{k,l,N}$ and $q_{k,l,N}=\bP\left(\bX(N) \in \cU_{k,l,N}\right)$ so that \eqref{eq:eqn42} can be written as a summation of the conditional expectations for all $\bX(N)=\bw \in\cU_{k,l,N}$ \cite{woods} after replacing $q_{k,l,N}$
\begin{align}
&\bP\left(b(N)=b_{l,N}\right)E[S(N) \vert b(N)=b_{l,N}] \nonumber \\
&=\sum_{b_{k,N-1}\in \cN_{l,N}}\sum_{\bw=\left[w_1\;w_2\right]^T\in \cU_{k,l,N}}E\left[S_{N}\vert b(N)=b_{l,N}, \right.\nonumber \\
& \left.b(N-1)=b_{k,N-1},\bX(N)=\bw\right]\nonumber \\
&\times \bP\left(b(N-1)=b_{k,N-1}\right)\bP\left(\bX(N)=\bw\vert \bX(N)\in\cU_{k,l,N}\right) \nonumber \\
&\times\bP\left(\bX(N) \in \cU_{k,l,N}\right)\label{eq:eqn43}.
\end{align} Now, given that $b(N-1)=b_{k,N-1}$, $b(N)=b_{l,N}$ and $\bX(N)=\bw=\left[w_1\;w_2\right]^T$, we observe that $\bP\left(\bX(N)=\bw\vert \bX(N)\in\cU_{k,l,N}\right)\bP\left(\bX(N) \in \cU_{k,l,N}\right)=\bP\left(\bX(N)=\bw\right)$ and
\begin{align}
&E\left[S_{N}\vert b(N)=b_{l,N},b(N-1)=b_{k,N-1},\bX(N)=\bw\right] \nonumber \\
&= E\left[S(N-1)(b_{k,N-1}w_1+(1-b_{k,N-1})w_2)\vert b(N-1)=\nonumber\right.\\ &\left.b_{k,N-1}\right], \label{eq:eqn45}
\end{align}
and by using \eqref{eq:eqn45} in \eqref{eq:eqn43}, we have
\begin{align*}
&\bP\left(b(N)=b_{l,N}\right)E[S(N) \vert b(N)=b_{l,N}] \nonumber \\
&=\sum_{b_{k,N-1}\in \cN_{l,N}}\sum_{\bw=\left[w_1\;w_2\right]^T\in \cU_{k,l,N}}E\left[S(N-1)(b_{k,N-1}w_1\right.\\
&\left.+(1-b_{k,N-1})w_2)\vert b(N-1)=b_{k,N-1}\right] \\
&\times\bP\left(b(N-1)=b_{k,N-1}\right)\bP\left(\bX(N)=\bw\right).
\end{align*} Therefore, we can write $\bP\left(b(N)=b_{l,N}\right)E[S(N)\vert b(N)=b_{l,N}]$ from $\bP\left(b(N-1)=b_{k,N-1}\right)E[S(N-1)\vert b(N-1)=b_{k,N-1}]$ as
\begin{align}
&\bP\left(b(N)=b_{l,N}\right)E[S(N) \vert b(N)=b_{l,N}] \nonumber \\
&=\sum_{b_{k,N-1}\in \cN_{l,N}} \bP\left(b(N-1)=b_{k,N-1}\right) \nonumber \\
&\times E\left[S(N-1)\vert b(N-1)=b_{k,N-1}\right] \nonumber \\ &\times\sum_{\bw=\left[w_1\;w_2\right]^T\in \cU_{k,l,N}}(b_{k,N-1}w_1+(1-b_{k,N-1})w_2)\nonumber\\&p_1(w_1)p_2(w_2) \label{eq:eqn44}
\end{align} for $l=2,\ldots,M_N$, where we use $\bP\left(\bX(N)=\bw\right)=p_1(w_1)p_2(w_2)$.

In the second case, if the portfolio $b(N)=b_{1,N}$, then there are two sets of price relative vectors that connect $b_{k,N-1}$ to $b_{1,N}$, i.e., $\cV_{k,1,N}$ and $\cR_{k,1,N}$. Depending on the value of the price vector, the portfolio may be rebalanced to $b_{1,N}=b$. If $\bX(N)\in\cV_{k,1,N}$, then the portfolio is not rebalanced and no transaction fee is paid. If $\bX(N)\in\cR_{k,1,N}$, then the portfolio is rebalanced and transaction cost is paid. We can find $\bP\left(b(N)=b_{1,N}\right)E[S(N) \vert b(N)=b_{1,N}]$ from $\bP\left(b(N-1)=b_{k,N-1}\right)E[S(N-1)\vert b(N-1)=b_{k,N-1}]$ as a summation of the conditional expectations for all $b_{k,N-1}\in\cN_{1,N}$ \cite{woods} as \vspace{-0.1in}
\begin{align}
&\bP\left(b(N)=b_{1,N}\right)E[S(N) \vert b(N)=b_{1,N}] \nonumber \\
&=\sum_{b_{k,N-1}\in \cN_{1,N}} E\left[S(N) \vert b(N)=b_{1,N},b(N-1)=b_{k,N-1}\right] \nonumber\\ & \times \bP\left(b(N-1)=b_{k,N-1}\right)q_{k,l,N}\label{eq:eqn51}.
\end{align} We note that given $b(N-1)=b_{k,N-1}$ and $b(N)=b_{1,N}$, the price relative vector $\bX(N)$ can take values from $\cV_{k,1,N}$ or $\cR_{k,1,N}$, $q_{k,l,N}=\bP\left(\bX(N) \in \cU_{k,l,N}\right)$ and $\bP\left(\bX(N)=\bw\vert \bX(N)\in\cU_{k,l,N}\right)\bP\left(\bX(N) \in \cU_{k,l,N}\right)=\bP\left(\bX(N)=\bw\right)$ which yields in \eqref{eq:eqn51} that \vspace{-0.1in}
\begin{align}
&\bP\left(b(N)=b_{1,N}\right)E[S(N) \vert b(N)=b_{1,N}] \nonumber \\
&=\sum_{b_{k,N-1}\in \cN_{l,N}}\left\{\sum_{\bw=\left[w_1\;w_2\right]^T\in \cV_{k,1,N}}E\left[S_{N}\vert b(N)=b_{l,N},\right.\right.\nonumber\\
&\left.\left. b(N-1)=b_{k,N-1},\bX(N)=\bw\right] \right.\nonumber \\
&\left. \times \bP\left(b(N-1)=b_{k,N-1}\right)\bP\left(\bX(N)=\bw\right) \right.\nonumber\\
&+ \sum_{\bw=\left[w_1\;w_2\right]^T\in \cR_{k,1,N}}E[S_{N}\vert b(N)=b_{l,N},b(N-1)=b_{k,N-1},\nonumber \\ & \left.\vphantom{\sum_{\bw=\left[w_1\;w_2\right]^T\in \cV_{k,1,N}}}\bX(N)=\bw]\bP\left(b(N-1)=b_{k,N-1}\right)\bP\left(\bX(N)=\bw\right)\right\}. \nonumber
\end{align}
If $\bX(N)=\bw\in\cV_{k,1,N}$, then it follows that
\begin{align}
&E\left[S_{N}\vert b(N)=b_{1,N},b(N-1)=b_{k,N-1},\bX(N)=\bw\right]  \nonumber \\
&= E\left[S(N-1)(b_{k,N-1}w_1 \right. \nonumber \\ & \left. +(1-b_{k,N-1})w_2)\vert b(N-1)=b_{k,N-1}\right]. \label{eq:eqn54}
\end{align}
If $\bX(N)=\bw\in\cR_{k,1,N}$, then transaction cost is paid which results
\begin{align}
&E\left[S_{N}\vert b(N)=b_{1,N},b(N-1)=b_{k,N-1},\bX(N)=\bw\right]  \nonumber\\
&=E\left[S(N-1)(b_{k,N-1}w_1+(1-b_{k,N-1})) \left(1-c \vphantom{\frac{b_{k,N-1}w_1}{b_{k,N-1}w_1+(1-b_{k,N-1})w_2}} \right.\right.\nonumber \\ 
& \left.\left. \times \left| \frac{b_{k,N-1}w_1}{b_{k,N-1}w_1+(1-b_{k,N-1})w_2}-b\right|\right)  {\big |} b(N-1)=b_{k,N-1}\right]. \label{eq:eqn55}
\end{align} Hence, we can write \eqref{eq:eqn51} after using \eqref{eq:eqn54} and \eqref{eq:eqn55} as
\begin{align}
&\bP\left(b(N)=b_{1,N}\right)E[S(N) \vert b(N)=b_{1,N}]\nonumber \\
&=\sum_{b_{k,N-1}\in \cN_{1,N}}\bP\left(b(N-1)=b_{k,N-1}\right)\nonumber \\
&\times\left\{ \sum_{\bw=\left[w_1\;w_2\right]^T\in \cV_{k,1,N}}\bP\left(\bX(N)=\bw\right)E\left[ S(N-1)\right.\right.\nonumber\\&\left.\left.\times(b_{k,N-1}w_1+(1-b_{k,N-1})w_2)\vert b(N-1)=b_{k,N-1}\right] \right. \nonumber\\
&+ \left. \sum_{\bw=\left[w_1\;w_2\right]^T\in \cR_{k,1,N}} \bP\left(\bX(N)=\bw\right) \right. \label{eq:eqn52}\\
&\left. \times E\left[S(N-1)(b_{k,N-1}w_1+(1-b_{k,N-1}))\right.\right.\nonumber\\&\left.\left.\left(1-c\left| \frac{b_{k,N-1}w_1}{b_{k,N-1}w_1+(1-b_{k,N-1})w_2}-b\right|\right)  \big | \right.\right.\nonumber\\&\left.\left.b(N-1)=b_{k,N-1}\right] \vphantom{\sum_{\bw=\left[w_1\;w_2\right]^T\in \cV_{k,1,N}}} \right\}. \nonumber
\end{align} Thus, we can write $\bP\left(b(N)=b_{1,N}\right)E[S(N) \vert b(N)=b_{1,N}]$ from $\bP\left(b(N-1)=b_{k,N-1}\right)E[S(N-1)\vert b(N-1)=b_{k,N-1}]$ as
\begin{align}
&\bP\left(b(N)=b_{1,N}\right)E[S(N) \vert b(N)=b_{1,N}]\nonumber \\
&=\sum_{b_{k,N-1}\in \cN_{1,N}} \bP\left(b(N-1)=b_{k,N-1}\right)E\left[S(N-1)\vert \right. \nonumber \\
&\left. b(N-1)=b_{k,N-1}\right] \left\{ \sum_{\bw=\left[w_1\;w_2\right]^T\in \cV_{k,1,N}}(b_{k,N-1}w_1 \right.\nonumber\\&\left. +(1-b_{k,N-1})w_2)p_1(w_1)p_2(w_2) \right. \label{eq:eqn53}\\  &+ \left. \vphantom{} \sum_{\bw=\left[w_1\;w_2\right]^T\in \cR_{k,1,N}}(b_{k,N-1}w_1+(1-b_{k,N-1}))\right.\nonumber\\&\left.\times\left(1-c\left| \frac{b_{k,N-1}w_1}{b_{k,N-1}w_1+(1-b_{k,N-1})w_2}-b\right|\right) \right.\nonumber\\&\left.\times p_1(w_1)p_2(w_2)\right\}, \nonumber
\end{align} which yields the recursive expressions for $\bP\left(b(N)=b_{l,N}\right)E[S(N) \vert b(N)=b_{l,N}]$ iteratively for each $l=1,\ldots,M_N$ with \eqref{eq:eqn44} and \eqref{eq:eqn53}.

Hence, we can calculate $E\left[S(N)\vert
  b(N)=b_{l,N}\right]$ $\times\bP\left(b(N)=b_{l,N}\right)$ for the case 
where the portfolio $b(N)=b_{l,N}$ for
$l=2,\ldots M_{N}$ by \eqref{eq:eqn44} and for the case where
the portfolio $b(N)=b_{1,N}=b$ by \eqref{eq:eqn53}. Therefore, we can
evaluate $E[S(N)]$ iteratively by \eqref{eq:eqn2}. Since, we have the
recursive formulation, we can optimize $b$ and $\epsilon$ by a brute
force search as shown in the Simulations section. For this recursive
evaluation, we have to find the set of achievable portfolios at each
investment period to compute $E[S(N)]$ by \eqref{eq:eqn2}. Hence, we next
analyze the number of calculations required to evaluate the expected
achieved wealth $E[S(N)]$.

\subsubsection{Complexity Analysis of the Iterative Algorithm}
We next investigate the number of achievable
portfolios at a given market period to determine the complexity of the
iterative algorithm. We show that the set of achievable portfolios at
period $N$ is equivalent to the set of achievable portfolios when the
portfolio $b(n)$ does not leave the interval $(b-\eps,b+\eps)$ for $N$
investment periods. We first demonstrate that if the portfolio never
leaves the interval $(b-\eps,b+\eps)$ for $N$ periods, then $b(N)$ is
given by $b(N)= 1/ \left( 1+\frac{1-b}{b}e^{\sum_{n=1}^{N} Z(n)} \right)$,
where $Z(n) \defi \ln\frac{X_2(n)}{X_1(n)}$ with a sample space $\cZ=\left\{z=\ln\frac{u}{v}\;\vert\;u,v\in\cX\right\}$ where $\vert \cZ \vert = M$. Then, we argue that the number of achievable portfolios at period $N$, $M_N$, is equal to the number of different values that the sum $\sum_{n=1}^{N}Z(n)$ can take when the portfolio does not leave the interval $(b-\eps,b+\eps)$ for $N$ investment periods.  We point out that $M\leq K^2-K+1$ since the price relative sequences $X_1(n)$ and $X_2(n)$ are elements of the same sample space $\cX$ with $\vert \cX \vert = K$ and by using this, we find an upper bound on the number of achievable portfolios.

\begin{lemma}
The number of achievable portfolios at period $N$, $M_N$, is equal to the number of different values that the sum $\sum_{n=1}^{N}Z(n)$ can take when the portfolio $b(n)$ does not leave the interval $(b-\eps,b+\eps)$ for $N$ investment periods and is bounded by $\binom{N+K^2-K}{N}$, i.e., $M_N = |\cB_N| \leq \binom{N+K^2-K}{N}$.
\label{lem:num_of_states}
\end{lemma}

\begin{proof}The proof is in the Appendix A. \end{proof}

\begin{remark}
Note that the complexity of calculating $E[S(N)]$ is bounded by $\cO
\left(\sum_{n=1}^{N} \binom{n+K^2-K}{n}/N\right)$ since at each period
$n=1,\ldots,N$, we calculate $E[S(n)]$ as a summation of $M_n$ terms,
i.e., $E[S(n)]=\sum_{l=1}^{M_n}E[S(n)\vert
  b(n)=b_{l,n}]\bP\left(b(n)=b_{l,n}\right)$ and $M_n \leq
\binom{n+K^2-K}{n}$.
\end{remark}

In the next section, we extend the given iterative algorithm to calculate the expected achieved wealth in a market
with $m$-assets, where $m$ is an arbitrary number determined by the investor. \vspace{-0.16in}

\subsection{Generalization of the Iterative Algorithm to the $m$-asset Market Case\label{subsec:m-asset}}
In this section, we generalize the iterative method introduced in Section~\ref{subsec:algorithm} to
a market with $m$ assets where $m \in \mathbbm{Z}^+$. We model the market as a sequence of i.i.d. price relative vectors
$\bX(n) = [X_1(n) \; X_2(n) \ldots X_m(n)]$, where $X_i(n) \in \cX$ and the p.m.f. of $X_i(n)$ is
$p_i(x) \defi \bP(X_i(n)=x)$. For $m$-asset case, the portfolio vector is given by
$\bb(n) = [b_1(n) \; b_2(n) \ldots b_m(n)]$,  $\sum_{i} b_i(n) = 1$, $b_i(n) \geq 0$, target portfolio vector is defined as $\bb=[b_1 \; b_2 \ldots b_m]$
and the threshold vector is given by $\beps = [\eps_1 \; \eps_2 \ldots \eps_m]$. Along these lines,
TRP$(\bb,\beps)$ rebalances the wealth allocation $\bb(n)$ to $\bb$ only when
$\bb(n) \notin \bb^\beps \defi [b_1-\eps_1, b_1+\eps_1]\times[b_2-\eps_2, b_2+\eps_2]\times \ldots \times[b_m-\eps_m, b_m+\eps_m]$.
In this case, if the wealth allocation is not rebalanced for $N$ investment periods, then the proportion of wealth invested in the $i$th
asset becomes
$
b_i(N) = \frac{b_i\prod_{n=1}^{N} X_i(N)}{\sum_{k=1}^m b_k\prod_{n=1}^N X_k(N)}
$
and achieved wealth is given by
$
S(N) = \sum_{k=1}^m b_k\prod_{n=1}^N X_k(N).
$
We define the set of achievable portfolios at period $N$ as
\begin{align*}
\cB_N=&\left\{\bb_{1,N}, \bb_{2,N}, \ldots, \bb_{M_N,N} \; \vert \; \bb_{k,N} = \frac{\bb_{l,N-1} \circ \bx}{\bx^T \bb_{l,N-1}} \in \bb^\beps \; \right.\\
 &\left.\mathrm{ or } \; \bb_{k,N} = \bb, \,\, \bx \in \cX^m \right\}
\end{align*}
where $M_N = |\cB_N|$. In accordance with the definitions given in two-asset market case, the definitions of the portfolio transition sets and the transition
probabilities of achievable portfolios follows. Then similar to the iterative algorithm introduced in Section~\ref{subsec:algorithm}, \eqref{eq:eqn44} and \eqref{eq:eqn53}, we can evaluate  $\bP\left(\bb(N)=\bb_{l,N}\right)E[S(N)\vert
 \bb(N)=\bb_{l,N}]$ for $l=1,\ldots,M_N$ from
$\bP\left(\bb(N-1)=\bb_{k,N-1}\right)E[S(N-1)\vert \bb(N-1)=\bb_{k,N-1}]$ for
$k=1,\ldots,M_{N-1}$ iteratively to calculate $E[S(N)]$. Therefore for the $m$-asset market case, by using
\begin{align*}
E[S(N)] = \sum_{l=1}^{M_N} \bP\left(\bb(N)=\bb_{l,N}\right) E\left[S(N) | \bb(N)=\bb_{l,N} \right],
\end{align*}
the expected achieved wealth $E[S(N)]$ can be evaluated iteratively.

In the next section, we show that the set of all achievable
portfolios, $\cB \defi \cup_{n=1}^{\infty} \cB_n$, is finite under
mild technical conditions.  \vspace{-0.16in}

\subsection{Finitely Many Achievable Portfolios}\label{subsec:finite}
In this section, we investigate the cardinality of the set of
achievable portfolios $\cB$ and demonstrate that $\cB$ is finite under
certain conditions in the following theorem,
Theorem~\ref{delta2}. This result is significant since when $\cB$ is
finite, we can derive a recursive update with a constant complexity,
i.e., the number of states does not grow, to calculate the expected
achieved wealth $E[S(n)]$ at any investment period. We demonstrate
that the portfolio sequence forms a Markov chain with a finite state
space and converges to a stationary distribution. Then, we can
investigate the limiting behavior of the expected achieved wealth
using this update to optimize $b$ and $\eps$. Before providing the
main theorem, we first state a couple of lemmas that are used in the
derivation of the main result of this section.

We first point out that in Lemma~\ref{lem:num_of_states}, we showed
that the number of achievable portfolios at period $N$ is equal to the
number of different values that the sum $\sum_{n=1}^{N}Z(n)$ can take
when the portfolio $b(n)$ does not leave the interval
$(b-\eps,b+\eps)$ for $N$ investment periods. Then, we observed that
the cardinality of the set $\cB$ is equal to the number of different
values that the sum $\sum_{n=1}^{N}Z(n)$ can take for any
$N\in\mathbb{N}$ when the portfolio $b(n)$ never leaves the interval
$(b-\eps,b+\eps)$. We next show that the portfolio $b(n)$ does not
leave the interval $(b-\eps,b+\eps)$ for $N$ periods if and only if
the sum $\sum_{n=1}^k Z(n) \in (\alpha_2,\alpha_1)$ for $k=1,\ldots,N$,
where $\alpha_1\defi\ln\frac{b(1-b+\eps)}{(1-b)(b-\eps)}>0$ and
$\alpha_2\defi\ln\frac{b(1-b-\eps)}{(1-b)(b+\eps)}<0$. Moreover, we
also prove that the number of achievable portfolios is equal to the
cardinality of the set $\cM\cap (\alpha_2,\alpha_1)$ where we define
the set $\cM$ as
\begin{align*}
	\cM = \{ &m_1 z_1 + m_2 z_2 + \ldots + m_{M^+} z_{M^+} \;\vert\; m_i \in \mathbb{Z},\;z_i\in\cZ^+\;\\
	&\mathrm{for}\; i=1,\ldots,M^+\},
\end{align*} $\cZ^+ \defi \{ z\in\cZ \;\vert\;z\geq0 \}$, $M^+\defi\vert\cZ^+\vert$. Note that $\cZ^+$ is the set of positive elements of the set $\cZ$ and any value that the sum $\sum_{n=1}^N Z(n)$ can take is an element of $\cM$. Hence, if we can demonstrate that the set $\cM\cap (\alpha_2,\alpha_1)$ is finite under certain conditions, then it yields the cardinality of the set $\cB$ since $\cB$ is finite if and only if $M\cap (\alpha_2,\alpha_1)$ is finite.

In the following lemma, we prove that the portfolio $b(n)$ does not
leave the interval $(b-\eps,b+\eps)$ for $N$ periods if and only if
the sum $\sum_{n=1}^k Z(n) \in (\alpha_2,\alpha_1)$ for
$k=1,\ldots,N$.

\begin{lemma} \label{lem:sumofZs}
The portfolio $b(n)$ does not leave the interval $(b-\eps,b+\eps)$ for $N$ investment periods if and only if the sum $\sum_{n=1}^{k}Z(n)\in(\alpha_2,\alpha_1)$ for $k=1,\ldots,N$.
\end{lemma}

\begin{proof}The proof is in the Appendix B. \end{proof}

In the following lemma, we demonstrate that if the condition $\vert z
\vert<\mathrm{min}\{\vert \alpha_1\vert,\vert\alpha_2\vert\}$ is
satisfied for each $z\in\cZ^+$, then for any element $m\in\cM \cap
(\alpha_2,\alpha_1)$, there exists an $N$-period market scenario where
the portfolio does not leave the interval $(b-\eps,b+\eps)$ for $N$
investment periods and $\{Z(n)=Z^{(n)}\}_{n=1}^N$ such that
$m=\sum_{n=1}^{N} Z^{(n)}$ for some $\{Z^{(n)}\}_{n=1}^N \in \cZ$ and
$N\in\mathbb{N}$. It follows that the set of different values that the
sum $\sum_{n=1}^{N} Z(n)$ can take for any $N\in\mathbb{N}$ when the
portfolio never leaves the interval $(b-\eps,b+\eps)$ for $N$
investment periods is equivalent to the set $\cM \cap
(\alpha_2,\alpha_1)$. Hence, we show that the cardinality of the set
of achievable portfolios is equal to the cardinality of the set $\cM
\cap (\alpha_2,\alpha_1)$. After this lemma, we present conditions
under which the set $\cM\cap(\alpha_2,\alpha_1)$ is finite so that the
set of achievable portfolios is also finite.

\begin{lemma} \label{lemma:M_achieve}
If $\vert z \vert<\mathrm{min}\{\vert \alpha_1\vert,\vert\alpha_2\vert\}$ for $z\in\cZ^+$, then any element of $\cM \cap (\alpha_2,\alpha_1)$ can be written as a sum $\sum_{n=1}^N Z^{(n)}$ for some $N\in\mathbb{N}$ where $\{Z(n)=Z^{(n)}\}_{n=1}^N\in\cZ$ and $\sum_{n=1}^k Z^{(n)}\in(\alpha_2,\alpha_1)$ for $k=1,\ldots,N$.
\end{lemma}

\begin{proof}
In Lemma~\ref{lem:num_of_states}, we showed that for any investment period $N$, the number of different portfolio values that $b(N)$ can take is equal to the number of different values that the sum $\sum_{n=1}^N Z(n)$ can take where $\sum_{n=1}^k Z(n)\in(\alpha_2,\alpha_1)$ for $k=1,\ldots,N$. Since this is true for any investment period $N$, it follows that the number of all achievable portfolios is equal to the number of different values that the sum $\sum_{n=1}^{N} Z(n)$ can take for any $N\in\mathbb{N}$ such that $\sum_{n=1}^{N} Z(n)\in(\alpha_2,\alpha_1)$.

Here, we show that if $m\in\cM \cap (\alpha_2,\alpha_1) $, then there exists a sequence $\{Z^{(n)}\}_{n=1}^N \in \cZ$ for some $N\in\mathbb{N}$ such that $m=\sum_{n=1}^{N} Z^{(n)}$ and $\sum_{n=1}^{k} Z^{(n)}\in(\alpha_2,\alpha_1)$ for $k=1,\ldots,N$. Let $m\in\cM \cap (\alpha_2,\alpha_1)$. Then, it can be written as $m=m_1z_1 + \ldots + m_{M^+}z_{M^+}$ for some $m_i\in\mathbb{Z}$ and $z_i\in\cZ^+$, $i=1,\ldots,M^+$. We define $S(k) = \sum_{n=1}^{k} Z^{(n)}$ for $k\geq1$ and construct a sequence $\{Z^{(n)}\}_{n=1}^N \in \cZ$ for some $N\in\mathbb{N}$ such that $m=\sum_{n=1}^{N} Z^{(n)}$  and $S(k)\in(\alpha_2,\alpha_1)$ for each $k=1,\ldots,N$ as follows. We choose $z_i\in\cZ^+$ such that $m_i>0$, let $Z^{(1)}=z_i$ and decrease $m_i$ by 1. We see that $S(1)=Z^{(1)}\in(\alpha_2,\alpha_1)$ since $z_i<\mathrm{min}\{\vert \alpha_1\vert,\vert\alpha_2\vert\}$. Next, we choose $z_j\in\cZ^+$ such that $m_j<0$, let $Z^{(2)}=-z_j$ and increase $m_j$ by 1. Then, it follows that $S(2)=Z^{(1)}+Z^{(2)}=z_i-z_j\in (\alpha_2,\alpha_1)$ since $z_i,z_j<\mathrm{min}\{\vert \alpha_1\vert,\vert\alpha_2\vert\}$. At any time $k\geq3$, if
\begin{itemize}
\item $S(k)\geq0$, we choose $z_l\in\cZ^+$ such that $m_l<0$, let $Z^{(k+1)}=-z_l$ and increase $m_l$ by 1. Note that $S(k+1) \in(\alpha_2,\alpha_1)$ since $S(k)\in(\alpha_2,\alpha_1)$, $S(k)\geq0$ and $Z^{(k+1)}<0$. Now assume that there
exists no $z_l\in\cZ^+$ such that $m_l<0$, i.e., $m_j\geq0$ for $j=1,\ldots,M$. If we let $I\defi\{j\in\{1,\ldots,M\}\;\vert\;m_j\geq0\}=\{k_1,\ldots,k_T\}$ where $T\defi\vert I \vert$ and $Z^{(l)}=z_{k_j},\;\;\; l=k+1+\sum_{i=1}^{j-1}k_i,\ldots,k+\sum_{i=1}^{j}k_i$ for $j=1,\ldots,T$, then we get that $m=S(N)=\sum_{n=1}^{N} Z^{(n)}$ where $N=k+\sum_{i=1}^T k_i$. We observe that $S_i\in(\alpha_2,\alpha_1)$ for $i=k+1,\ldots,N$ since $m\in(\alpha_2,\alpha_1)$, $\sum_{j=1}^{T}m_{k_j}x_{k_j}\geq0$ and $S(k)>0$.

\item $S(k)<0$, we choose $z_l\in\cZ^+$ such that $m_l>0$, let $Z^{(k+1)}=z_l$ and decrease $m_l$ by 1. Note that $S(k+1) \in(\alpha_2,\alpha_1)$ since $S(k)\in(\alpha_2,\alpha_1)$, $S(k)<0$ and $Z^{(k+1)}\geq0$. Assume that there
exists no $z_l\in\cZ^+$ such that $m_l\geq0$, i.e., $m_j<0$ for $j=1,\ldots,M$. If we let $J\defi\{j\in\{1,\ldots,M\}\;\vert\;m_j\leq0\}=\{k_1,\ldots,k_W\}$ where $W\defi\vert J \vert$ and $Z^{(l)}=z_{k_j},\;\;\; l=k+1+\sum_{i=1}^{j-1}k_i,\ldots,k+\sum_{i=1}^{j}k_i$ for $j=1,\ldots,W$, then we get that $m=S(N)=\sum_{n=1}^{N} Z^{(n)}$ where $N=k+\sum_{i=1}^W k_i$. We see that $S_i\in(\alpha_2,\alpha_1)$ for $i=k+1,\ldots,N$ since $m\in(\alpha_2,\alpha_1)$, $\sum_{j=1}^{W}m_{k_j}x_{k_j}\leq0$ and $S(k)<0$.
\end{itemize}
Therefore, we can write $m=\sum_{n=1}^N Z^{(n)}$ for some $N\geq1$ where $\{Z^{(n)}\}_{n=1}^N \in \cZ$ and $\sum_{n=1}^k Z^{(n)} \in(\alpha_2,\alpha_1)$ for $k=1,\ldots,N$.
\end{proof}

Hence, we showed that if the condition $\vert z
\vert<\mathrm{min}\{\vert \alpha_1\vert,\vert\alpha_2\vert\}$ is
satisfied for each $z\in\cZ^+$, then any element of the set $\cM \cap
(\alpha_2,\alpha_1)$ can be written as a sum $\sum_{n=1}^{N} Z(n)$ for
some $N\in\mathbb{N}$ when the portfolio does not leave the interval
$(b-\eps,b+\eps)$ for $N$ investment periods. It follows that the set
of different values that the sum $\sum_{n=1}^{N} Z(n)$ can take for
any $N\in\mathbb{N}$ when the portfolio does not leave the interval
$(b-\eps,b+\eps)$ for $N$ investment periods is equivalent to the set
$\cM \cap (\alpha_2,\alpha_1)$. Thus, the number of achievable
portfolios is equal to the cardinality of the set $\cM \cap
(\alpha_2,\alpha_1)$. In the following theorem, we demonstrate that if
$\vert z \vert<\mathrm{min}\{\vert \alpha_1\vert,\vert\alpha_2\vert\}$
for $z\in\cZ^+$ and the set $\cM$ has a minimum positive element, then
$\cM \cap (\alpha_2,\alpha_1)$ is finite. Hence, the set of achievable
portfolios is also finite under these conditions. Otherwise, we show
that the set $\cM \cap (\alpha_2,\alpha_1)$ contains infinitely many
elements so that the set of achievable portfolios is also
infinite. Thus, we show that the set of achievable portfolios is
finite if and only if the minimum positive element of the set $\cM$
exists.

\begin{theorem}
If $\vert z \vert<\mathrm{min}\{\vert \alpha_1\vert,\vert\alpha_2\vert\}$ for $z\in\cZ^+$ and the set $\cM$ has a minimum positive element, i.e., if
\begin{equation*}
\delta = \min \{ m\in \cM \;|\; m > 0 \}
\label{delta2}
\end{equation*}
exists, then the set of achievable portfolio $\cB = \cup_{n=1}^{\infty} \cB_n$ is finite. If such a minimum positive element does not exist, then $\cB$ is countably infinite.
\label{thm:num_of_states}
\end{theorem}

In Theorem~\ref{thm:num_of_states} we present a necessary and
sufficient condition for the achievable portfolios to be finite. We
emphasize that the required condition, i.e., $\vert z
\vert<\mathrm{min}\{\vert \alpha_1\vert,\vert\alpha_2\vert\}$ for
$z\in\cZ^+$, is a necessary required technical condition which assures
that the TRP thresholds are large enough to prohibit constant
rebalancings at each investment period. In this sense, this condition does not limit the generality of the TRP framework.

By Theorem~\ref{thm:num_of_states}, we establish the conditions for a
unique stationary distribution of the achievable portfolios. With the
existence of a unique stationary distribution, in the next section, we
provide the asymptotic behavior of the expected wealth growth by
presenting the growth rate.

\begin{proof}
For any investment period $N$, we showed in Lemma~\ref{lem:num_of_states} that the number of different portfolio values that $b(N)$ can take is equal to the number of different values that the sum $\sum_{n=1}^{N} Z(n)$ can take where the sum $\sum_{n=1}^{k} Z(n)\in(\alpha_2,\alpha_1)$ for $k=1,\ldots,N$. In the Lemma~\ref{lemma:M_achieve}, we showed that the set of different values that the sum $\sum_{n=1}^{N} Z(n)$ can take where the sum $\sum_{n=1}^{k} Z(n)\in(\alpha_2,\alpha_1)$ for $k=1,\ldots,N$ is equivalent to the set $\cM \cap (\alpha_2,\alpha_1)$. We let $\cH$ be the set of values that the sum $\sum_{n=1}^{N} Z(n)\in(\alpha_2,\alpha_1)$ can take for any $N\in\mathbb{N}$, i.e., $\cH=\{\sum_{n=1}^{N} Z^{(n)}\;\vert\;\{Z^{(n)}\}_{n=1}^N\in \cZ,\; \sum_{n=1}^{k} Z^{(n)}\in(\alpha_2,\alpha_1)\mathrm{ for }k=1,\ldots,N,\;N\in\mathbb{N}\}$. Now, assume that the minimum positive element $\delta$ exists. We next illustrate that the sum $\sum_{n=1}^{N} Z^{(n)}$ for any sequence $\{Z^{(n)}\}_{n=1}^N\in \cZ$ can be written as $k\delta$ for some $k\in\mathbb{Z}$, i.e., $\sum_{n=1}^{N} Z^{(n)}=k\delta$.

Assume that there exists a sequence $\{Z^{(n)}\}_{n=1}^N\in \cZ$ such that the sum $Z=\sum_{n=1}^{N} Z^{(n)} \not = k\delta$ for any $k\in\mathbb{Z}$. If we divide the real line into intervals of length $\delta$, then $Z$ should lie in one of the intervals, i.e., there exists $k_0\in\mathbb{Z}$ such that $k_0\delta<Z<(k_0+1)\delta$ so that we can write $Z=k_0\delta + \eta$ where $0<\eta<\delta$. By definition of $\cM$, an integer multiple of any element of $\cM$ is also an element of $\cM$ so that $k_0\delta\in\cM$ since $\delta\in\cM$. Moreover, for any two elements of $\cM$, their difference is also an element of $\cM$ so that $\eta=Z- k_0\delta \in \cM$ since $Z\in\cM$ and $k_0\delta\in\cM$. However, this contradicts to the fact that $\delta$ is the minimum positive element of $\cM$ since $0<\eta<\delta$ and $\eta\in\cM$. Hence, it follows that any element of $\cH$ can be written as $k\delta$ for some $k\in\mathbb{Z}$. Note that there are finitely many elements in $\cH$ since any element $h\in \cH$ can be written as $h=k\delta$ for some $k\in\mathbb{Z}$ and $\alpha_2<h<\alpha_1$. Since $\vert \cB \vert = \vert \cH \vert$, it follows that the set of achievable portfolios $\cB$ is finite.

To show that if $\delta$ does not exist then $\cB$ contains infinitely many elements, we assume that $\delta$ does not exist. Since every finite set of real numbers has a minimum, there are either countably infinitely many positive elements in the set $\cM$ or none. We know that there exists $z_i \neq 0$ so that there are positive numbers in $\cM$. Therefore, there are infinitely many elements in $\cM$. Now assume that there exists $\gamma_1>0$ that can be written as a sum $\sum_{n=1}^N Z^{(n)}$ for some $N\in\mathbb{N}$ where $\{Z^{(n)}\}_{n=1}^N\in \cZ$ and $\sum_{n=1}^k Z^{(n)}\in(\alpha_2,\alpha_1)$. Then, by Lemma~\ref{lemma:M_achieve}, it follows that $\gamma_1\in \cM \cap (0,\alpha_1)$ and since there exists no positive minimum element of $\cM$, there exists $\gamma_2>0$ such that $\gamma_2<\gamma_1$ so that $\gamma_2\in \cM \cap (0,\alpha_1)$. In this way, we can construct a decreasing sequence $\{\gamma_n\}$ such that $\gamma_n\in\cM \cap (0,\alpha_1)$ for each $n\in\mathbb{N}$. Note that for any $n\in\mathbb{N}$, $\gamma_n$ is also element of $\cH$ by Lemma~\ref{lemma:M_achieve} so that there are countably infinite elements in $\cH$. Hence, it follows that $\cB$ has countably infinitely many elements.
\end{proof}

We showed that if $\vert z \vert<\mathrm{min}\{\vert
\alpha_1\vert,\vert\alpha_2\vert\}$ for $z\in\cZ^+$ and the minimum
positive element of the set $\cM$ exists, then the set of achievable
portfolios, $\cB$, is finite. If the minimum positive element of the
set $\cM$ does not exist, then the set $\cM \cap (\alpha_2,\alpha_1)$
is countably infinite so that the number of achievable portfolios is
also countably infinite. Hence, the set of achievable portfolios is
finite if and only if the minimum positive element of the set $\cM$
exists. However, Theorem~\ref{thm:num_of_states} does not specify the
exact number of achievable portfolios. In the following corollary, we
demonstrate that the number of achievable portfolios is
$\lfloor\frac{\alpha_1 - \alpha_2}{\delta}\rfloor$ if the set of
achievable portfolios is finite.

\begin{corollary}
If $\vert z \vert<\mathrm{min}\{\vert \alpha_1\vert,\vert\alpha_2\vert\}$ for $z\in\cZ^+$ and $
\delta = \min \{ m | m > 0 \, m \in \cM \}$ exists, then the number of achievable portfolios is$\lfloor\frac{\alpha_1 - \alpha_2}{\delta}\rfloor$\footnote{Here, $\lfloor x/y\rfloor$ is the largest integer less than or equal to $x/y$}.
\label{lem:ach_of_states}
\end{corollary}

\begin{proof}
Assume that $\delta$ exists and there exists $\theta>0$ such that $\theta$ can be written as a sum $\sum_{n=1}^N Z^{(n)}$ for some $N\in\mathbb{N}$ and $\{Z(n)=Z^{(n)}\}_{n=1}^N\in\cZ$ such that $\sum_{n=1}^k Z^{(n)} \in (\alpha_2,\alpha_1)$ for $k=1,\ldots,N$. Note that such a $\theta$ exists, e.g., $\theta=z>0$ where $z\in\cZ^+$ since $z\in(\alpha_2,\alpha_1)$. Then, by Lemma~\ref{lemma:M_achieve}, it follows that $\theta \in \cM \cap (0,\alpha_1)$. Since $\delta$ is the minimum positive element of $\cM$, it follows that $0<\delta\leq\theta$ and $\delta\in \cM\cap(0,\alpha_1)$.
Hence, by Lemma~\ref{lemma:M_achieve}, we get that $\delta$ can be written as a sum $\sum_{n=1}^{N^{'}} Z^{(n)}$ for some $N^{'}\in\mathbb{N}$ and $\{Z^{(n)}\}_{n=1}^{N^{'}}\in\cZ$ where $\sum_{n=1}^k Z^{(n)} \in (\alpha_2,\alpha_1)$ for $k=1,\ldots,N{'}$. We note that $\delta$ is an element of the set of different values that the sum $\sum_{n=1}^N Z(n)$ can take for any $N\in\mathbb{N}$ and $Z(n)\in\cZ$ for $n=1,\ldots,N$ such that the portfolio does not leave the interval $(b-\eps,b+\eps)$. We showed in Theorem~\ref{thm:num_of_states} that any element of $\cM$ can be written as $k\delta$ for some $k\in\mathbb{Z}$ so that the number of elements in $\cM \cap (\alpha_2,\alpha_1)$ is $\lfloor\frac{\alpha_1 - \alpha_2}{\delta}\rfloor$. Hence, it follows that there are exactly $\lfloor\frac{\alpha_1 - \alpha_2}{\delta}\rfloor$ achievable portfolios since Lemma~\ref{lemma:M_achieve} implies that the set $\cM \cap (\alpha_2,\alpha_1)$ is equivalent to the set of different values that the sum $\sum_{n=1}^N Z(n)$ can take for any $N\in\mathbb{N}$ and $Z(n)\in\cZ$ for $n=1,\ldots,N$ such that the sum $\sum_{n=1}^k Z(n) \in (\alpha_2,\alpha_1)$ for each $k=1,\ldots,N$ and the cardinality of the latter set is equal to the number of achievable portfolios.
\end{proof}

In Theorem~\ref{thm:num_of_states}, we introduce conditions on the cardinality of the set of all achievable portfolio states, $\cB$, and showed that if $\vert z \vert<\mathrm{min}\{\vert \alpha_1\vert,\vert\alpha_2\vert\}$ for all $z\in\cZ^+$ and the minimum positive element of the set $\cM$ exists, then $\cB$ is finite. This result is significant when we analyze the asymptotic behavior of the expected achieved wealth, i.e., in the following, we demonstrate that when $\cB$ is finite, the portfolio sequence converges to a stationary distribution. Hence, we can determine the limiting behavior of the expected achieved wealth so that we can optimize $b$ and $\eps$. To accomplish this, specifically, we first present a recursive update to evaluate $E[S(n)]$. We then maximize $g(b,\eps) \defi \limn \frac{1}{n} \log E[S(n)]$ over $b$ and $\eps$ with a brute-force search, i.e., we calculate $g(b,\eps)$ for different $(b,\eps)$ pairs and find the one that yields the maximum. \vspace{-0.16in}
\subsection{Finite State Markov Chain for Threshold Portfolios \label{subsec:markov}}
If we assume that $\vert z \vert<\mathrm{min}\{\vert \alpha_1\vert,\vert\alpha_2\vert\}$ for all $z\in\cZ^+$ and $\delta=\mathrm{min}\{m\in\cM\;\vert\;m>0\}$ exists, then the set of all achievable portfolios $\cB$ is finite. By Corollary~\ref{lem:ach_of_states}, it follows that there are exactly $L=\lfloor\frac{\alpha_1 - \alpha_2}{\delta}\rfloor$ achievable portfolios. We let $\cB=\{b_1,\ldots,b_L\}$ and, without loss of generality, $b_1=b$. We define the probability mass vector of the portfolio sequence as $\vpi(n)=\left[\pi_1(n)\;\ldots\;\pi_L(n)\right]^T$ where $\pi_i(n)\defi \bP\left(b(n)=b_i\right)$. The portfolio sequence $b(n)$ forms a homogeneous Markov chain with a finite state space $\cB$ since the transition probabilities between states are independent of period $n$. We see that $b(n)$ is irreducible since each state communicates with other states so that all states are null-persistent since $\cB$ is finite \cite{woods}. Then, it follows that there exists a unique stationary distribution vector $\vpi$, i.e., $\vpi = \limn \vpi(n)$. To calculate $\vpi$, we first observe that the set of portfolios that are connected to $b_l$, $\cN_{l,n}$, and the set of price relative vectors that connect $b_k$ to $b_l$, $\cU_{k,l,n}$, are independent of investment period since the price relative sequences are i.i.d. for $k=1,\ldots,L$ and $l=1,\ldots,L$. Hence, we write $\cU_{k,l,n}=\cU_{k,l}$ and $\cN_{l,n}=\cN_l$ for $n\in\mathbb{N}$. We next note that the state transition probabilities are also independent of investment period and write $q_{k,l,n}=\bP\left(b(n)=b_l\vert b(n-1)=b_k\right) = q_{k,l}$ for $n\in\mathbb{N}$, $k=1,\ldots,L$ and $l=1,\ldots,L$. Therefore, we can write $\bP\left(b(n)=b_l\right)$ as
\begin{align}
\bP\left(b(n)=b_l\right)= \sum_{k=1}^L q_{k,l}\bP\left(b(n-1)=b_k\right), \label{eq:pir1}
\end{align} where $q_{k,l}=0$ if $b_k\not\in\cN_{l}$. Now, by using the definition of $\vpi(n)$ and \eqref{eq:pir1}, we get $
\vpi(n+1) = \mP \vpi(n)$ for each $n$, where $\mP$ is the state transition matrix, i.e., $\mP_{ij}=q_{i,j}$.

We next determine the limiting behavior of the expected achieved wealth $E[S(n)]$ to optimize $b$ and $\eps$ as follows. In Section~\ref{subsec:algorithm}, we showed that $E[S(n)]$ can be calculated iteratively by \eqref{eq:eqn2}, \eqref{eq:eqn44} and \eqref{eq:eqn53}. If we define the vector $\bee(n)=\left[e_{1}(n)\;\ldots\;e_{L}(n)\right]^T$ where $e_i(n)\defi\bP\left(b(n)=b_i\right)E[S(n)\vert b(n)=b_i]$, then we can calculate $E[S(n)]$ as the sum of the entries of $\bee(n)$ by \eqref{eq:eqn2}, i.e.,
$
E[S(n)] = \sum_{i=1}^L e_i(n) = \mathbf{1}^T \bee(n)$, where $\mathbf{1}$ is the vector of ones. Hence, by definition of $\bee(n)$, we can write
$
\bee(n+1) = \bQ \bee(n), $
where the matrix $\bQ$ is given by
\begin{align}
&\bQ=\label{eq:Q}\\
&\left[\begin{matrix}
\sum\limits_{\bw=\left[w_1\;w_2\right]^T\in \cU_{1,1}} \kappa_1 & \cdots & \sum\limits_{\bw=\left[w_1\;w_2\right]^T\in \cU_{L,1}} \kappa_L  \\
\vdots  & \ddots & \vdots \\
\sum\limits_{\bw=\left[w_1\;w_2\right]^T\in \cU_{1,L}} \kappa_1 & \cdots & \sum\limits_{\bw=\left[w_1\;w_2\right]^T\in \cU_{L,L}} \kappa_L
\end{matrix}\right] \nonumber
\end{align}
where $\kappa_i \defi \left(b_iw_1 +(1-b_i)w_2\right)p_1(w_1)p_2(w_2)$ and we ignore rebalancing for presentation purposes. From
\eqref{eq:eqn44} and \eqref{eq:eqn53}, $\bQ$ does not depend on period
$n$ since there are finitely many portfolio states, i.e., $\bQ$ is
constant. If we take rebalancing into account, then only the first row
of the matrix $\bQ$ changes and the other rows remain the same where
\begin{align*}
\bQ_{1,j} &= \sum_{\bw=\left[w_1\;w_2\right]^T\in \cV_{j,1}} \left(b_1w_1 +(1-b_1)w_2\right)p_1(w_1)p_2(w_2) \nonumber \\
          &+ \sum_{\bw=\left[w_1\;w_2\right]^T\in \cR_{j,1}} \left(b_1w_1 +(1-b_1)w_2\right) \\
          & \times \left(1-c \left|\frac{b_1w_1}{b_1w_1+(1-b_1)w_2}-b\right|\right)p_1(w_1)p_2(w_2),
\end{align*}
$\cV_{j,1}$ is the set of price relative vectors that connect $b_j$ to $b_1=b$ without crossing the threshold boundaries and
$\cR_{j,1}$ is the set of price relative vectors that connect $b_j$ to $b_1=b$ by crossing the threshold boundaries for $i=j,\ldots,L$. Note that we can find the matrix $\bQ$ by using the set of achievable portfolios $\cB$ and the probability mass vectors $\bp_1$ and $\bp_2$ of the price relative sequences.

Here, we analyze $E[S(n)]$ as $n\rightarrow\infty$ as follows. We
assume that the matrix $\bQ$ is diagonalizable with the eigenvalues
$\lambda_1,\ldots,\lambda_L$ and, without loss of generality,
$\lambda_1\geq\ldots\geq\lambda_L$, which is the case for a wide range
of transaction costs \cite{woods}. Then, there exists a
nonsingular matrix $\bB$ such that $\bQ=\bB \bLambda \bB^{-1}$ where
$\bLambda$ is the diagonal matrix with entries
$\lambda_1,\ldots,\lambda_L$. We observe that the matrix $\bQ$ has
nonnegative entries. Therefore, it follows from Perron-Frobenius
Theorem \cite{meyer} that the matrix $\bQ$ has a unique largest
eigenvalue $\lambda_1>0$ and any other eigenvalue is strictly smaller
than $\lambda_1$ in absolute value, i.e., $\lambda_1>\vert \lambda_j
\vert$ for $j=2,\ldots,L$. Then, the recursion on $\vec{e}(n)$ yields
$ \bee(n) = \bQ^n \bee(0)$, and after some algebra, 
 the expected achieved wealth $E[S(n)]$ is given by
\begin{align*}
E[S(n)] = \mathbf{1}^T \bee(n) = \mathbf{1}^T \bB  \bLambda^n \bB^{-1} \bee(0) = \sum_{i=1}^L u_i v_i \lambda_i^n,
\end{align*}where $\bu\defi\left[u_1\;\ldots\;u_L\right]^T=\bB^{T} \mathbf{1}$ and $\bv\defi\left[v_1\;\ldots\;v_L\right]=\bB^{-1}\bee(0)$. Then, it follows that
\begin{align*}
g(b,\eps) &= \limn \frac{1}{n} \log E[S(n)] = \limn \frac{1}{n}\log \left\{\sum_{i=1}^L u_i v_i \lambda_i^n\right\} \\
                                          &= \limn \log \lambda_1 + \limn \frac{1}{n} \log \left\{\sum_{i=1}^L u_i v_i \left(\frac{\lambda_i}{\lambda_1}\right)^n \right\} \\
                                          &= \log \lambda_1
\end{align*} since $\limn \left(\frac{\lambda_i}{\lambda_1}\right)^n = 0 $ for $i=2,\ldots,L$. Hence, we can optimize $b$ and $\eps$ as
\begin{align*}
\left[b^*,\eps^*\right] = \argmax_{b\in[0,1],0<\eps} g(b,\eps)= \argmax_{b\in[0,1],0<\eps}\log \lambda_1.
\end{align*} To maximize $g(b,\eps)$, we evaluate it for different values of $(b,\eps)$ pairs and find the pair that maximizes $g(b,\eps)$, i.e., by a brute-force search in the Simulations section.

In the next section, we investigate the well-studied two-asset
Brownian market model with transaction costs.  \vspace{-0.2in}

\subsection{Two Stock Brownian Markets \label{subsec:brownian}}
In this section, we consider the well-known two-asset Brownian market,
where stock price signals are generated from a standard Brownian
motion \cite{iyengargo,davis1990,taksar}. Portfolio selection problem
in continuous time two-asset Brownian markets with proportional
transaction costs was investigated in \cite{taksar}, where the growth
optimal investment strategy is shown to be a threshold
portfolio. Here, as usually done in the financial literature
\cite{davis1990}, we first convert the continuous time Brownian market
by sampling to a discrete-time market \cite{iyengargo}. Then, we
calculate the expected achieved wealth and optimize $b$ and $\eps$ to
find the best portfolio rebalancing strategy for a discrete-time
Brownian market with transaction costs.  Note that although, the
growth optimal investment in discrete-time two-asset Brownian markets
with proportional transaction costs was investigated in
\cite{iyengargo}, the expected achieved wealth and the optimal
threshold interval $(b-\eps,b+\eps)$ has not been calculated yet.

To model the Brownian two-asset market, we use the price relative
vector $\bX = \left[X_1\;X_2\right]^T$ with $X_1 = 1$ and $X_2 =
e^{kZ}$ where $k$ is constant and $Z$ is a random variable with
$\bP\left(Z=\pm1\right)=\frac{1}{2}$. This price relative vector is
obtained by sampling the stock price processes of the continuous time
two-asset Brownian market \cite{iyengargo,taksar}. We emphasize that
this sampling results a discrete-time market  identical to the
binomial model popular in asset pricing \cite{iyengargo}. We first
present the set of achievable portfolios and the transition
probabilities between portfolio states. 

Since the price of the first stock is the same over investment
periods, the portfolio leaves the interval $(b-\eps,b+\eps)$ if either
the money in the second stock grows over a certain limit or  falls
below a certain limit. If the portfolio $b(n)$ does not leave the
interval $(b-\eps,b+\eps)$ for $N$ investment periods, then the money
in the first stock is $b$ dollars and the money in the second stock is
$(1-b)e^{ki}$ for some $-N\leq i\leq N$ so that the portfolio is
$b(N)=\frac{b}{b+(1-b)e^{ki}}$. Note that $\frac{b}{b+(1-b)e^{ki}}\in
(b-\eps,b+\eps)$ if and only if $i_{\mmin}\leq i \leq i_{\mmax}$, where
$i_{\mmin} \defi \left\lceil\frac{1}{k} \ln
\frac{b(1-b-\eps)}{(1-b)(b+\eps)}\right\rceil$\footnote{Here, $\lceil x/y \rceil$ is the largest integer greater or equal to the $x/y$.} and $i_{\mmax} \defi
\left\lfloor{\frac{1}{k} \ln
  \frac{b(1-b+\eps)}{(1-b)(b-\eps)}}\right\rfloor$. Hence, the set of
achievable portfolios is given by
\begin{align*}
\cS = &\left\{b_i = \frac{b}{b + (1-b) e^{({i+i_{\mmin}-1})k}} \; \vert \; i = 1,\ldots, \right. \\ & \left. \vphantom{\frac{b}{b + (1-b) e^{({i+i_{\mmin}-1})k}}} 
i_{\mmax} - i_{\mmin} +1 \right\}= \{b_1,\ldots,b_S\},
\end{align*} where $\vert \cS \vert =S$ and $S\defi i_{\mmax} - i_{\mmin} +1$ and $b_{1-i_{\mmin}}=b$. We see that the portfolio is rebalanced to $b_{1-i_{\mmin}}=b$ only if it is in the state $b_1$ and $ X_2= e^{-k}$ or if it is in the state $b_S$ and $X_2=e^{k}$. Therefore, the transition probabilities are given by
\begin{align*}
&\bP\left(b_i|b_j\right) \\
&= \left\{
	\begin{array}{ll}
       \frac{1}{2} \;:\; i=2,\ldots,S-1\;\mathrm{and}\;j = i\pm 1\;\;, \mathrm{or} \;\;i=1 \;\;\mathrm{and}\\ j\in\{2,1-i_{\mmin}\},\;\;\mathrm{or}\;\;i=S\;\mathrm{and}\;j\in\{S-1,1-i_{\mmin}\} \\
       0 \,\,:\,\, \mathrm{otherwise},
     \end{array}
   \right.
\end{align*} where $P(b_i|b_j)$ is the probability that the portfolio $b(n)=b_i$ given that $b(n-1)=b_j$ for any period $n$. We now calculate $E[S(n)]$ using the recursions in Section~\ref{subsec:markov} as follows. The sets of price relative vectors that connect portfolio states are given by
\begin{align*}
\cU_{i,j} = \left\{
\begin{array}{ll}
\{\left[1\;e^{k}\right]^T\}\;&:\; i=1,\ldots,S-1\;\mathrm{and}\; j=i+1,\;\;\\ &\mathrm{or}\;\; i=S\;\mathrm{and}\;j=1-i_{\mmin}\\
\{\left[1\;e^{-k}\right]^T\}\;&:\; i=2,\ldots,S-1\;\mathrm{and}\; j=i-1,\;\;\\ &\mathrm{or}\;\; i=1\;\mathrm{and}\;j=1-i_{\mmin}\\
\varnothing\;&:\;\mathrm{otherwise}.
\end{array}\right.
\end{align*} Hence, we can calculate the matrix $\bQ$ defined in \eqref{eq:Q} as
\begin{align*}
\bQ_{i,j} = \left\{
\begin{array}{ll}
\frac{1}{2}(b_{j} +(1-b_{j})e^{k})\;&:\; i=2,\ldots,S\;\\ &\mathrm{and}\; j=i-1 \\
\frac{1}{2}(b_{j} +(1-b_{j})e^{-k})\;&:\; i=1,\ldots,S-1\;\\& \mathrm{and}\; j=i+1 \\
0 \;&:\;\mathrm{otherwise},
\end{array}\right.
\end{align*} where we ignore rebalancing. If we take rebalancing into account, then
$
\bQ_{1-i_{\mmin},1}=\frac{1}{2}(b_1 + (1-b_1)e^{-k}) \left(1-c\left|\frac{b_1}{b_1+(1-b_1)e^{-k}}-b\right|\right)
$ and
$\bQ_{1-i_{\mmin},S}=\frac{1}{2}(b_S + (1-b_S)e^{k})\left(1-c\left|\frac{b_S}{b_S+(1-b_S)e^{k}}-b\right|\right).$ Then, by the recursions in  Section~\ref{subsec:markov}, $E[S(n)]$ is given by $\bQ^n\bee(0)$. Moreover, we maximize
$
g(b,\eps) = \limn \frac{1}{n} \log E[S(n)]= \log \lambda_1,
$ where $\lambda_1$ is the largest eigenvalue of the matrix $\bQ$. Here, we optimize $b$ and $\eps$ with a brute-force search, i.e., we find $\lambda_1$ for different $(b,\eps)$ pairs and find the one that achieves the maximum.
%

 \vspace{-0.1in}

\section{Maximum Likelihood Estimators of The Probability Mass Vectors} \label{sec:ML}
In this section, we sequentially estimate the probability mass vectors
$\bp_1$ and $\bp_2$ corresponding to $X_1(n)$ and $X_2(n)$,
respectively, using a maximum likelihood estimator (MLE). In general,
these vectors may not be known or change in time, hence, could be
estimated at each investment period prior to calculation of
$E[S(n)]$. The maximum likelihood estimator for a pmf on a finite set
is well-known \cite{woods}, but we provide the corresponding
derivations here for completeness. We consider, without loss of
generality, the price relative sequence $X_1(n)$ and assume that its
realizations are given by $X_1(n)=w_n\in\cX$ for $n=1,\ldots,N$ and
estimate $\bp_1$. Similar derivations follow for the price relative
sequence $X_2(n)$ and $\bp_2$. Note that as demonstrated in the
Simulations section, the corresponding estimation can be carried out
over a finite length window to emphasize the most recent data.  We
define the realization vector $\bw=\left[w_1,\ldots,w_N\right]$ and
the probability mass function as
$p_{\bth}(x_i)=p_1(x_i\vert\bth)=\theta_{x_i}$ for $i=1,\ldots,K$ and
the parameter vector $\bth\defi
\left[\theta_{x_1},\ldots,\theta_{x_K}\right]$. Then, the MLE of the
probability mass vector $\bp_1$ is given by
$
\bth_{\mathrm{MLE}} = \argmax_{\bth:\sum_{i=1}^K \theta_{x_i}=1} p_1(\bw\vert\bth) = \argmax_{\bth:\sum_{i=1}^K \theta_{x_i}=1} \bP\left(X_1(1)=w_1,\ldots,X_1(N)=w_N\vert\bth\right)$. Since the price relative sequence $X_1(n)$ is i.i.d., it follows that
$
p_1(\bw\vert\bth) = \prod_{i=1}^N p_1(w_i\vert\bth) = \prod_{i=1}^N \theta_{w_i}= \prod_{i=1}^N \prod_{j=1}^K \theta_{x_j}^{\mathrm{I}(w_i=x_j)}, $ since $\mathrm{I}(.)$ is the indicator function. If we change the order of the product operators, then we obtain
$
p_1(\bw\vert\bth) = \prod_{i=1}^N \prod_{j=1}^K \theta_{x_j}^{\mathrm{I}(w_i=x_j)} = \prod_{j=1}^K \theta_{x_j}^{\sum_{i=1}^N\mathrm{I}(w_i=x_j)}
                = \prod_{j=1}^K \theta_{x_j}^{N_j},$ where $N_j\defi \sum_{i=1}^N\mathrm{I}(w_i=x_j)$, i.e., the number of realizations that are equal to $x_j\in\cX$ for $j=1,\ldots,K$. Note that $\sum_{j=1}^K N_j = N$. Hence, we can write 
$
\bth_{\mathrm{MLE}} = \argmax_{\bth:\sum_{i=1}^K \theta_{x_i}=1} \sum_{j=1}^K \frac{N_j}{N} \log \theta_{x_j},$ since $\log(.)$ is a monotone increasing function. If we define the vector $\bh=\left[ h_{x_1},\ldots,h_{x_K}\right]$, where $h_{x_j}\defi \frac{N_j}{N}$ for $j=1,\ldots,K$, then we see that $h_{x_j}\geq0$ for $j=1,\ldots,K$ and $\sum_{j=1}^K h_{x_j}=1$. Since $\bh$ and $\bth$ are probability vectors, i.e., their entries are nonnegative and sum to one, it follows that $ \mathrm{D}(\bh\Vert\bth) \defi \sum_{i=1}^K h_{x_j} \log\left(
\frac{h_{x_j}}{\theta_{x_j}} \right) \geq 0 $ and  $\mathrm{D}(\bh\Vert\bth)=0$ if and only if $\bth=\bh$, i.e., their relative entropy is nonnegative \cite{coverthom}. Therefore, we get that
$\sum_{j=1}^K \frac{N_j}{N} \log \theta_{x_j} = \sum_{j=1}^K h_{x_j} \log \theta_{x_j} = -\mathrm{D}(\bh\Vert\bth) + \sum_{j=1}^K h_{x_j} \log h_{x_j} \leq \sum_{j=1}^K h_{x_j} \log h_{x_j}, $
 where the equality is reached if and only if $\bth=\bh$. Hence, it follows that
$
\bth_{\mathrm{MLE}} = \argmax_{\bth:\sum_{i=1}^K \theta_{x_i}=1} \sum_{j=1}^K \frac{N_j}{N} \log \theta_{x_j}= \bh
$ so that we estimate the probability mass vector $\bp_1$ with $\bh=\left[\frac{N_1}{N},\ldots,\frac{N_K}{N}\right]$ at each investment period $N$ where $\frac{N_j}{N}$ is the proportion of realizations up to period $N$ that are equal to $x_j$ for $x_j\in\cX$.  \vspace{-0.15in}

\section{Simulations} \label{sec:sim}
In this section, we demonstrate the performance of TRPs with several
different examples. We first analyze the performance of TRPs in a
discrete-time two-asset Brownian market introduced in
Section~\ref{subsec:brownian}. As the next example, we apply TRPs to
historical data from \cite{cover1991,KoSi11} collected from the New
York Stock Exchange over a 22-year period and compare the results to
those obtained from other investment strategies
\cite{KoSi11,cover1991,iyengar2005,port}.  We show that the
performance of the TRP algorithm is significantly better than the
portfolio investment strategies from
\cite{KoSi11,cover1991,iyengar2005,port} in historical data sets as
expected from Section III.

\begin{figure}[t]
\centering
    \subfloat[]{\centerline{\epsfxsize=9cm \epsfbox{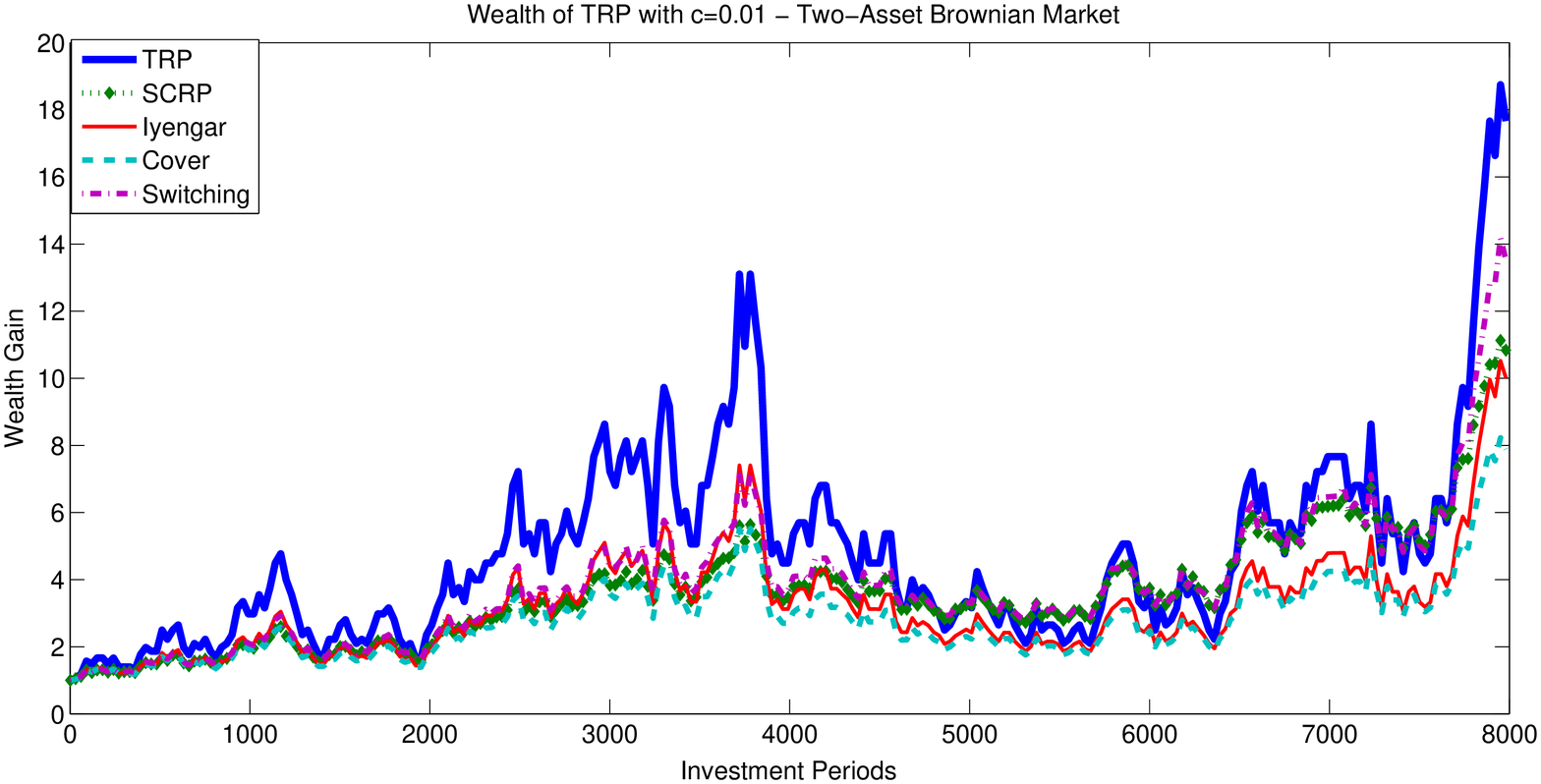}}}\\
    \subfloat[]{\centerline{\epsfxsize=9cm \epsfbox{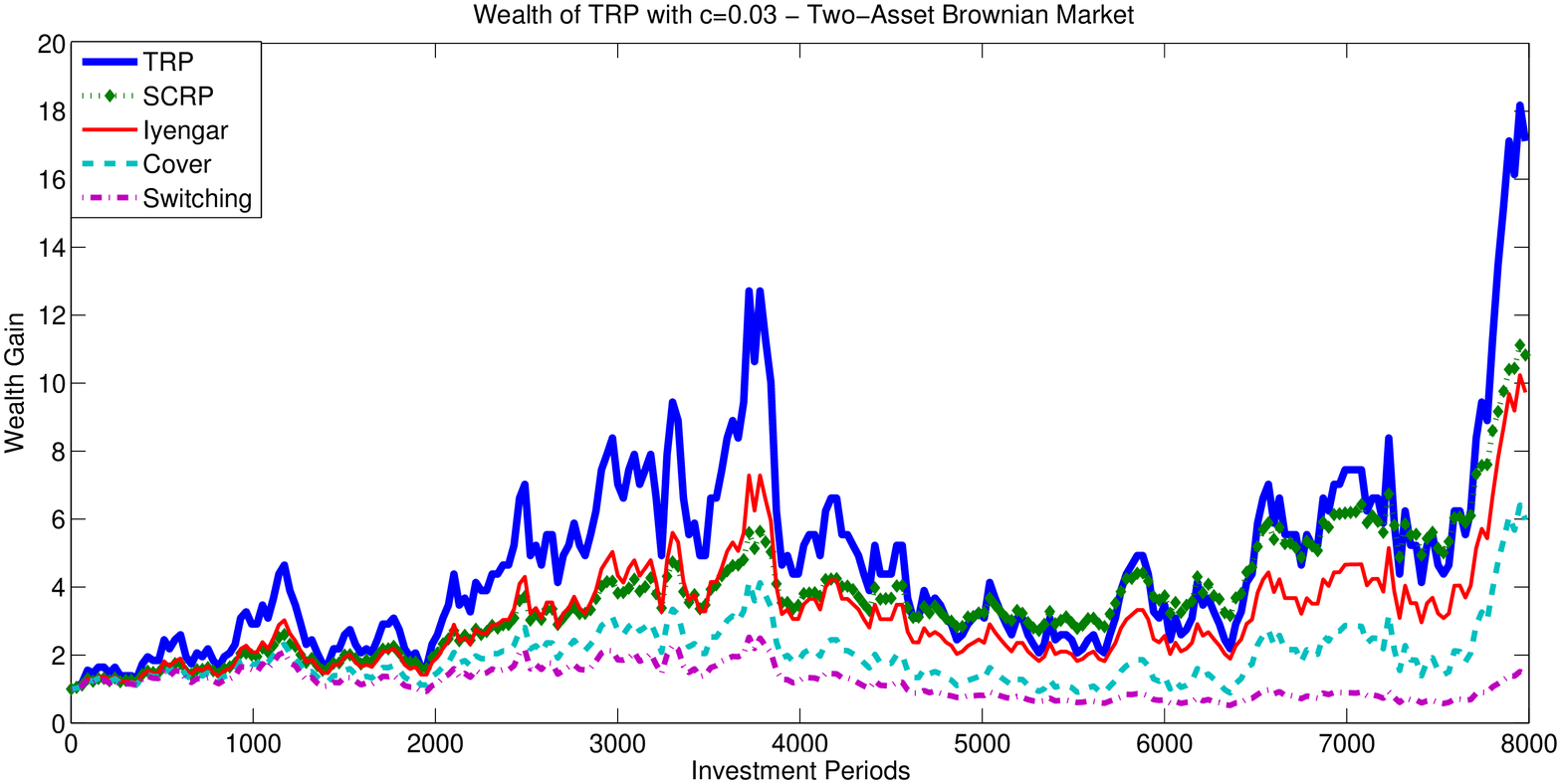}}}
\caption{Performance of portfolio investment strategies in the two-asset Brownian market. (a) Wealth gain with the cost ratio $c=0.01$. (b) Wealth gain with the cost ratio $c=0.03$.}
\label{fig:brownian}
\end{figure}

As the first scenario, we apply TRPs to a discrete-time two-asset
Brownian market. Under this well studied market in the financial
literature \cite{investment}, the price relative vector is given by
$\vec{X}=\left[X_1\; X_2\right]^T$, where $X_1=1$, $X_2=e^{kZ}$ and
$Z=\pm1$ with equal probabilities and we set $k=0.03$
\cite{iyengargo}. Here, the sample spaces of the price relative
sequences $X_1$ and $X_2$ are $\cX_1 = \left\{1\right\}$ and $\cX_2 =
\left\{0.97,1.03\right\}$, respectively, and $\cX = \cX_1 \cup \cX_2 =
\left\{x_1,x_2,x_3\right\}$, where $x_1=1$, $x_2=0.97$,
$x_3=1.03$. Hence, the probability mass vectors of the price relative
sequences $X_1$ and $X_2$ are given by $\bp_1=\left[1\;0\;0\right]^T$
and $\bp_2=\left[0\;0.5\;0.5\right]^T$, respectively. Based on this
data, we evaluate the growth rate for different $(b,\eps)$ pairs to
find the best TRP that maximizes the growth rate using the approach
introduced in Section~\ref{subsec:brownian}, i.e., we form the matrix
$\bQ$ and evaluate the corresponding maximum eigenvalues to find the
pair that achieves the largest maximum eigenvalue since this pair also
maximizes the growth rate. Then, we invest 1 dollars in a randomly
generated two-asset Brownian market using: the TRP, labeled as,
``TRP'', i.e., TRP($b$,$\eps$) with calculated $(b,\eps)$ pair, the
SCRP algorithm with the target portfolio vector
$\vb=\left[0.5\;0.5\right]$, labeled as ``SCRP'', as suggested in
\cite{KoSi11}, the Iyengar's algorithm, labeled as ``Iyengar'', the
Cover's algorithm, labeled as ``Cover'', and the switching portfolio,
labeled as ``Switching'', with parameters suggested in \cite{port}. In
Fig.~\ref{fig:brownian}, we plot the wealth achieved by each algorithm
for transaction costs $c=0.01$ and $c=0.03$, where $c$ is the
proportion paid when rebalancing, i.e, $c=0.03$ is a $3\%$
commission. As expected from the derivations in Section III, we
observe that, in both cases, the performance of the TRP algorithm is
significantly better than the other algorithms under transaction
costs.
\begin{figure}[t]
\centering
    \subfloat[]{\label{fig:average-a}\centerline{\epsfxsize=9cm \epsfbox{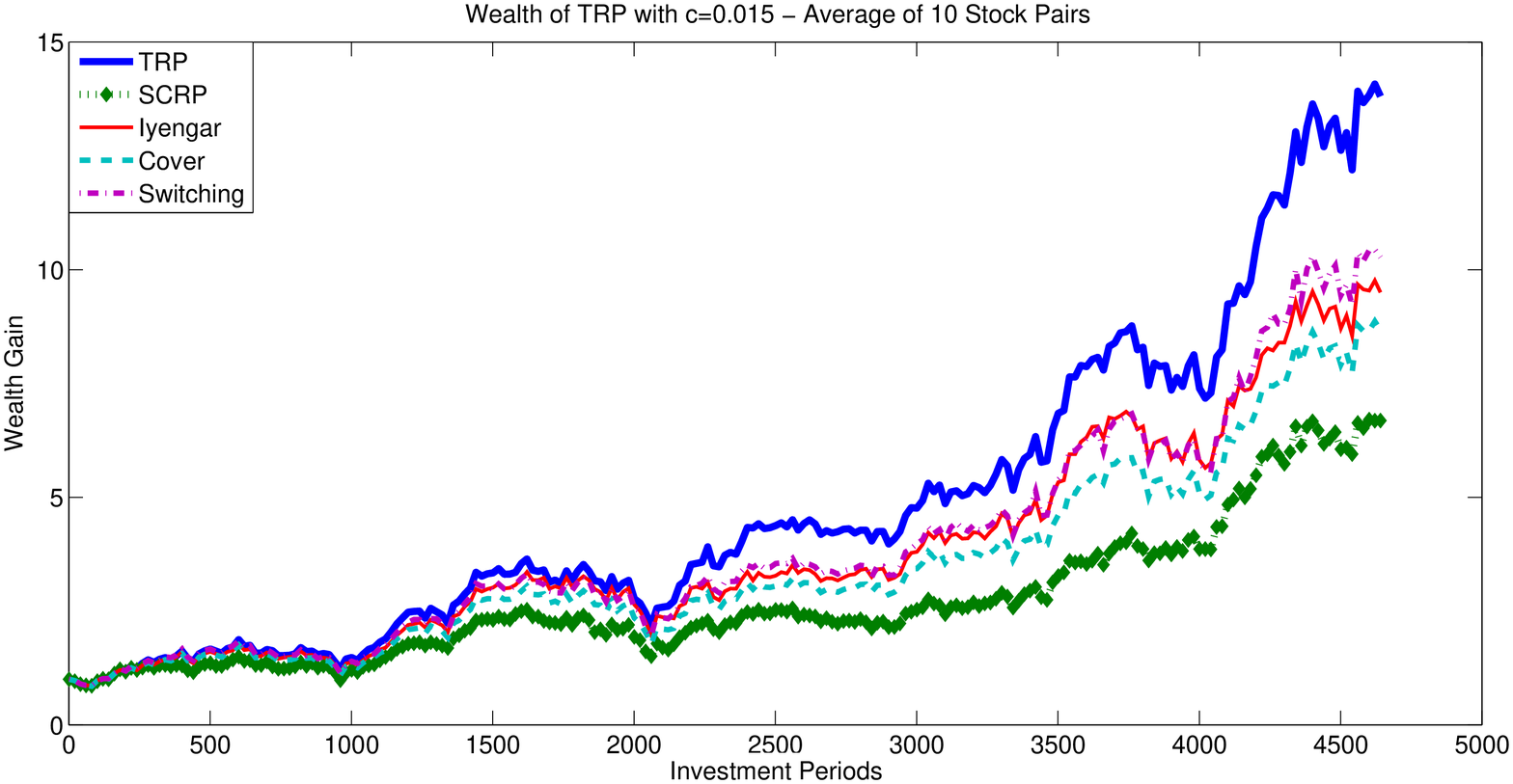}}}\\
    \subfloat[]{\label{fig:average-b}\centerline{\epsfxsize=9cm \epsfbox{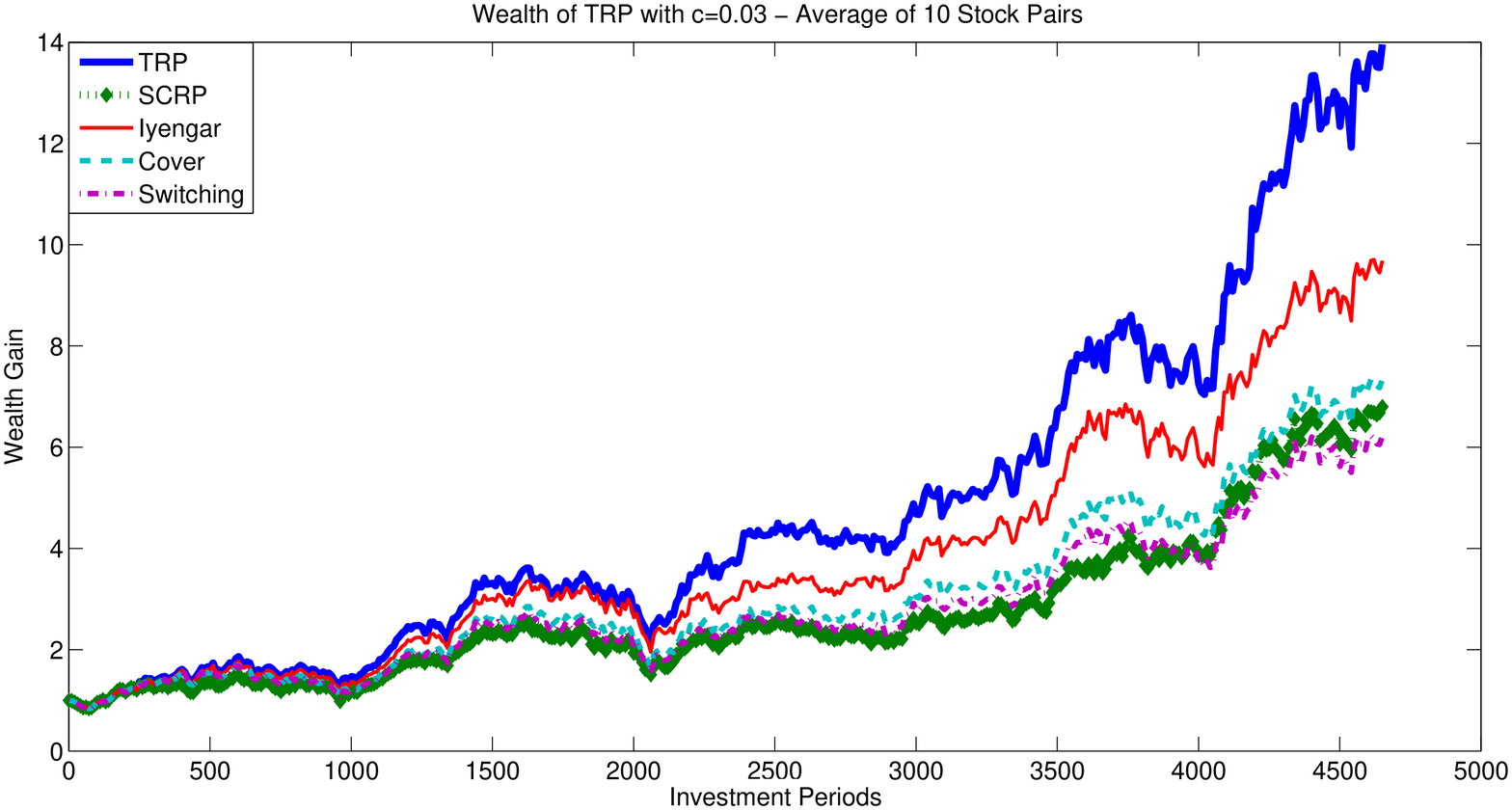}}}
\caption{Average performance of portfolio investment strategies on  independent stock pairs. (a) Wealth gain with the cost ratio $c=0.015$. (b) Wealth gain with the cost ratio $c=0.03$.}
\label{fig:average}
\end{figure}

We next present results that illustrate the average performance of TRPs 
on a number of stock pairs from the historical data sets \cite{cover1991} to 
avoid any bias to particular stock pairs. In this set of simulations, we first randomly
select pairs of stocks from the historical data that includes 34
stocks (where the Kin Ark stock is excluded). The data includes the price 
relative sequences of 34 stock pairs for 5651 investment periods (days). 
Since the brute force algorithm introduced in Section~\ref{subsec:algorithm} requires
the sample spaces of the price relative sequences, for each randomly
selected stock pair, we proceed as
follows. We first calculate the sample spaces and the probability mass
vectors of the price relative sequences from the first 1000-day
realizations of $X_1$ and $X_2$, where the sample spaces are simply
constructed by quantizing the observed realizations into bins. We
observed that the performance of the TRP is not effected by the number
of bins provided that there are an adequate number of bins to
approximate the continuous valued price relatives. Then, we optimize
$b$ and $\eps$ using the MLE introduced in Section IV and the brute
force algorithm from Section III, and invest using this TRP for the
next 1000 periods, i.e., from period 1001 to period 2000. We then
update $(b,\eps)$ pair using the first $2000$-day realizations of the
price relative vectors and invest using the best TRP for the next 1000
periods. We repeat this process through all available data. Hence, we
invest on each stock pairs using TRP for 4651 periods where we update
$(b,\eps)$ pair at each 1000 periods. In Fig.~\ref{fig:average-a} and 
Fig.~\ref{fig:average-b}, we
present the average performances of the TRP
algorithm, the SCRP algorithm, the Cover's algorithm, the Iyengar's
algorithm and the switching portfolio, under a mild transaction cost
$c=0.015$ and a hefty transaction cost $c=0.03$, respectively. In
Fig.~\ref{fig:average}, we present the wealth gain for each algorithm,
where the results are averaged over randomly selected 10 independent
stock pairs. We observe that although
the performance of the algorithms other than the TRP degrade with
increasing transaction cost, the performance of the TRP, using the
MLE, is not significantly effected since it can avoid excessive
rebalancings. We further observe from these simulations that the average
performance of the TRP is better than the average performance of the
other portfolio investment strategies commonly used in the
literature.
\vspace{-0.2in}

\section{Conclusion} \label{sec:con}
We studied growth optimal investment in i.i.d. discrete-time markets
under proportional transaction costs. Under this market model,
we studied threshold portfolios that are shown to yield the optimal
growth. We first introduced a recursive update to calculate the
expected growth for a two-asset market and then extend our results to
markets having more than two assets. We next demonstrated that under the threshold
rebalancing framework, the achievable set of portfolios form an
irreducible Markov chain under mild technical conditions. We evaluated
the corresponding stationary distribution of this Markov chain, which
provides a natural and efficient method to calculate the cumulative
expected wealth. Subsequently, the corresponding parameters are
optimized using a brute force approach yielding the growth optimal
investment portfolio under proportional transaction costs in
i.i.d. discrete-time two-asset markets. We also solved the optimal
portfolio selection in discrete-time markets constructed by sampling
continuous-time Brownian markets. For the case that the underlying
discrete distributions of the price relative vectors are unknown, we
provide a maximum likelihood estimator. We observed in our
simulations, which include simulations using the historical data sets
from \cite{cover1991}, that the introduced TRP algorithm significantly improves
the achieved wealth under both mild and hefty transaction costs as
predicted from our derivations.  \vspace{-0.2in}

\appendix

\noindent
{\bf A)} { Proof of Lemma~\ref{lem:num_of_states}:} We analyze the
cardinality of the set $\cB_N$ of achievable portfolios at period $N$,
$M_N$, as follows. If we assume that an investor invests with a
TRP($b$,$\eps$) for $N$ investment periods and the sequence of price
relative vectors are given by $\left\{
\left[X_1(n)\;X_2(n)\right]=\left[X_1^{(n)}\;X_2^{(n)}\right]
\right\}_{n=1}^N$ and the portfolio sequence is given by
$\left\{b(n)=b_{n}\right\}_{n=1}^N$, then we see that the portfolio
could leave the interval at any period depending on the realizations
of the price relative vector. We define an $N$-period market scenario
as a sequence of portfolios $\left\{b(n)\right\}_{n=1}^{N}$. We can
find the number of achievable portfolios at period $N$ as the number
of different values that the last element of $N$-period market
scenarios can take. Here, we partition the set of $N$-period market
scenarios according to the last time the portfolio leaves the interval
$(b-\eps,b+\eps)$ and show that any achievable portfolio at period $N$
can be achieved by an $N$-period market scenario where the portfolio
does no leave the interval $(b-\eps,b+\eps)$ for $N$ periods as
follows. If we define the set $\cP$ as the set of $N$-period market
scenarios, i.e.,
\begin{align*}
\cP  =\left\{\{b_{n}\}_{n=1}^N \;\vert\; b_{n}\in \cB_n\;,n=1,\ldots,N\right\}= \bigcup_{i=1}^{N+1} \cP_i,
\end{align*} where $\cP_i$ is the set of $N$-period market scenarios where the portfolio leaves the interval $(b-\eps,b+\eps)$ last time at period $i$, $i=1,\ldots,N$, and $\cP_{N+1}$ is the set of $N$-period market scenarios where the portfolio does not leave the interval $(b-\eps,b+\eps)$ for $N$ investment periods. We point out that $\cP_i$'s are disjoint, i.e., $\cP_i \cap \cP_j = \varnothing$ for $i\not=j$ and their union gives the set of all $N$-period market scenarios, i.e., $\bigcup_{i=1}^{N+1} \cP_i = \cP$ so that they form a partition for $\cP$. We see that the set $\cB_N$ of achievable portfolios at period $N$ is the set of last elements of $N$-period market scenarios, i.e., $\cB_N = \{b_N\; \vert\;\{b_n\}_{n=1}^N \in \cP\}$. We next show that the last element of any $N$-period market scenario from $\cP_i$ for $i=1,\ldots,N$ is also a last element of an $N$-period market scenario from $\cP_{N+1}$. Therefore, we demonstrate that any element of the set $\cB_N$ is achievable by a market scenario from $\cP_{N+1}$ and $\cB_N = \{b_N\; \vert\;\{b_n\}_{n=1}^N \in \cP_{N+1}\}$.

Assume that $\left\{b_{n}\right\}_{n=1}^N\in \cP_i$ for some $i\in\{1,\ldots,N\}$ so that $b_{i}=b$, i.e., the portfolio is rebalanced to $b$ last time at period $i$. Note that $b_{N}$ can also be achieved by an $N$-period market scenario $\left\{b^{'}_n\right\}_{n=1}^N$ where the portfolio never leaves the interval $(b-\eps,b+\eps)$, $b^{'}_j=b_{i+j}$ for $j=1,\ldots,N-i$ and $X_1^{(j)}=X_2^{(j)}$ for $j=N-i+1,\ldots,N$ so that $b^{'}_N=b^{'}_{N-i}=b_N$. Hence, it follows that the set of achievable portfolios at period $N$ is the set of achievable portfolios by $N$-period market scenarios from $\cP_{N+1}$. We next find the number of different values that $b(N)$ can take where the portfolio does not leave the interval $(b-\eps,b+\eps)$ for $N$ investment periods.

When the portfolio never leaves the interval $(b-\eps,b+\eps)$ for $N$ investment periods, $b(N)$ is given by
$
b(N) = b\prod_{i=1}^{N}X_1(n)/ \left(b\prod_{i=1}^{N}X_1(n)+(1-b)\prod_{i=1}^{N}X_2(n)\right)$. If we write the reciprocal of $b(N)$ as
$
\frac{1}{b(N)} = 1+\frac{1-b}{b}\prod_{n=1}^{N}\frac{X_2(n)}{X_1(n)}= 1+\frac{1-b}{b}e^{\sum_{n=1}^{N} Z(n)},
$ then we observe that the number of different values that the portfolio $b(N)$ can take is the same as the number of different values that the sum $\sum_{n=1}^{N} Z(n)$ can take. Since the price relative sequences $X_1(n)$ and $X_2(n)$ are elements of the same sample space $\cX$ with $\vert\cX\vert=K$, it follows that $\vert \cZ \vert = M \leq K^2-K+1$. Since the number of different values that the sum $\sum_{n=1}^{N} Z(n)$ can take is equal to $\binom{N+M-1}{M-1}$ and $M\leq K^2-K+1$, it follows that the number of achievable portfolios at period $N$ is bounded by $\binom{N+K^2-K}{K^2-K}$, i.e., $\vert \cB_N \vert = M_N \leq \binom{N+K^2-K}{K^2-K}$ and the proof follows. \\
\noindent
{\bf B)} { Proof of Lemma~\ref{lem:sumofZs}:}If the portfolio does not leave the interval $(b-\eps,b+\eps)$ for $N$ investment periods, then $b(n)\in(b-\eps,b+\eps)$ for $n=1,\ldots,N$ and it is not adjusted to $b$ at these periods so that
$
b(n) = \frac{b\prod_{i=1}^n X_1(i)}{b\prod_{i=1}^n X_1(i)+(1-b)\prod_{i=1}^n X_2(i)} \in (b-\eps,b+\eps)
$
for each $n=1,\ldots,N$. Taking the reciprocal of $b(n)$, we get that
$
\frac{b(1-b-\eps)}{(1-b)(b+\eps)}< \prod_{i=1}^n \frac{X_2(i)}{X_1(i)} < \frac{b(1-b+\eps)}{(1-b)(b-\eps)}.
$
Noting that $\frac{X_2(i)}{X_1(i)}=e^{Z(i)}$ and taking the logarithm of each side, it follows that
$
\ln\frac{b(1-b-\eps)}{(1-b)(b+\eps)}=\alpha_2< \sum_{i=1}^n Z(i) < \ln\frac{b(1-b+\eps)}{(1-b)(b-\eps)}=\alpha_1,
$
i.e., $\sum_{i=1}^n Z(i) \in(\alpha_2,\alpha_1)$ for $n=1,\ldots,N$. Now, if the portfolio leaves the interval $(b-\eps,b+\eps)$ first time at period $k$ for some $k\in\{1,\ldots,N\}$, then we get that $b(k)\geq b+\eps$ or $b(k) \leq b-\eps$
so that we get
$
\sum_{i=1}^k Z(i)\geq \alpha_1$ or $ \sum_{i=1}^k Z(i)\leq \alpha_2,
$
i.e., $\sum_{i=1}^k Z(i)\not \in (\alpha_2,\alpha_1)$. \vspace{-0.17in}

\bibliographystyle{IEEEbib}
{
\def\ninept{\def\baselinestretch{0.8}}
\ninept
\bibliography{msaf_references}}

\end{document}